\documentclass[12pt,a4paper]{article}
\usepackage[a4paper,hmargin={24mm,24mm},vmargin={27mm,27mm}]{geometry}
\usepackage{authblk}

\usepackage{mathrsfs}
\usepackage{amsfonts}
\usepackage{amssymb}
\usepackage{array}
\usepackage{eurosym}
\usepackage{bm}
\usepackage{multirow,bigstrut}
\usepackage{enumerate}
\usepackage{tikz}
\usepackage{graphicx}
\usepackage{stfloats}
\usepackage{float}
\usepackage{array}
\usepackage{booktabs}
\usepackage{amsmath}
\usepackage{titlesec}
\usepackage[colorlinks,
            linkcolor=red,
            anchorcolor=blue,
            citecolor=blue
            ]{hyperref}

\usepackage[linesnumbered,boxed,ruled,commentsnumbered]{algorithm2e}
\usepackage{algorithmicx}
\usepackage{algpseudocode}

\titleformat*{\section}{\large\bfseries}
\titleformat*{\subsection}{\normalsize\bfseries}

\newenvironment{proof}{\noindent{\em \textbf{Proof.}}}{\quad \hfill$\Box$\vspace{2ex}}

\newtheorem{theorem}{Theorem}[section]

\newtheorem{remark}[theorem]{Remark}

\newtheorem{lemma}[theorem]{Lemma}

\numberwithin{equation}{section}

\def \T {\mathbb{T}}

\def\qi {\bm{i}}

\def\sgn{\mbox{sgn}}

\def \F {\mathbf{F}}
\def \U {\mathbf{U}}

\newcommand{\norm}[1]{\left\lVert#1\right\rVert}
\newcommand{\abs}[1]{\left|#1\right|}

\title{\Large \textbf{Signal reconstruction from noisy multichannel samples}}

\author[a,b]{\small  Dong Cheng\thanks{chengdong720@163.com}}

\author[c]{\small  Xiaoxiao Hu\thanks{huxiaoxiao@wmu.edu.cn}}

\author[d]{\small Kit Ian Kou\thanks{kikou@umac.mo}}

\affil[a]{\small{Research Center for Mathematics and Mathematics Education, Beijing Normal University at Zhuhai, Zhuhai 519087, China}}
\affil[b]{\small{Laboratory of Mathematics and Complex Systems (Ministry of Education), School of Mathematical Sciences, Beijing Normal University, Beijing  100875, China}}
\affil[c]{\small{The First Affiliated Hospital of Wenzhou Medical University, Wenzhou Medical University, Wenzhou, Zhejiang, China}}
\affil[d]{\small{Department of Mathematics, Faculty of Science and Technology, University of Macau, Macao, China}}

\date{}

\begin{document}
  \maketitle
\begin{abstract}
\normalsize
We consider the signal reconstruction problem under the case of the signals  sampled in the multichannel way and with the presence of noise. Observing that if  the samples are inexact,  the rigorous enforcement of  multichannel interpolation    is inappropriate. Thus the reasonable  smoothing and regularized corrections are indispensable. In this paper,  we propose  several alternative methods for the signal reconstruction from the noisy multichannel samples under different  smoothing and regularization principles. We compare these  signal reconstruction  methods theoretically and experimentally in the various situations. To  demonstrate the effectiveness of the proposed methods, the probability interpretation and  the error analysis for these methods are  provided. Additionally, the numerical simulations as well as some guidelines to use the methods  are also presented.

\end{abstract}


\begin{msc}
 39A12,  11R52,	41A05, 12E05.
\end{msc}

\section{Introduction}\label{S1}

The main specialty of the multichannel sampling \cite{papoulis1977generalized,Grochenig2020sharp} is that the samples are taken from multiple transformed versions of the function. The transformation can be  the derivative, the Hilbert transform, or more general   liner   time invariant  system \cite{monich2017two}. The classical multichannel sampling theorem \cite{papoulis1977generalized} is only available for the bandlimited functions in the sense of Fourier transform and it  has been generalized for the bandlimited functions in the sense of   fractional Fourier transform (FrFT) \cite{liu2017signal},
linear canonical transform (LCT) \cite{xu2017multichannel,Shah2021Lattice} and offset LCT \cite{wei2019convolution}. In a real application, only finitely many samples, albeit with large amount, are given in a bounded region \cite{cheng2020Multicnon}. That is, the underlying signal is time-limited. Thus, reconstruction by the sampling formulas for the bandlimited functions is inappropriate because the bandlimited  functions cannot be time-limited by the uncertainty principle \cite{wigderson2021uncertainty}.  A time-limited function can be viewed as a  period of a  periodic function. Certain studies have been given to the sampling theorems for the periodic bandlimited functions \cite{xiao2013sampling,Mohammadi2018sampling}. Moreover,    the multichannel sampling approach has been extended to the time-limited functions \cite{cheng2019FFTMCI}.

Let $\mathbb{T}:=[0,2\pi)$ be the unit circle and denote by $L^p(\T),~1\leq q<\infty$, the totality of functions $f(t)$ such that
\begin{equation*}
\norm{f}_p:= \left(\frac{1}{2\pi} \int_{\T}|f(t)|^pdt\right)^{\frac{1}{p}}<\infty.
\end{equation*}
Let $M\in \mathbb{Z}^+$,  $f,h_m \in L^2(\T)$, and  define
\begin{equation*}
g_m(t)    = (f * h_m)(t) =  \frac{1}{2\pi} \int_{\T}f(s)h_m(t-s)ds,
\end{equation*}
for $1\leq m \leq M$. It was shown in  \cite{cheng2019FFTMCI}  that there exist  $y_1(t)$, $y_2(t)$, $\cdots$, $y_M(t)$ such that
\begin{equation}\label{MCI}
 \mathcal{T}_\mathbf{N} f(t) : = \frac{1}{L}\sum_{m=1}^{M} \sum_{p=0}^{L-1}g_m(\tfrac{2\pi p}{L})y_m(t-\tfrac{2\pi p}{L})
\end{equation}
 satisfies the following interpolation consistency:
\begin{equation}\label{consistency}
\left(\mathcal{T}_\mathbf{N}f*h_m\right)(\tfrac{2\pi p}{L}) = \left( f * h_m \right) (\tfrac{2\pi p}{L}),\quad  0\leq p\leq L-1,\quad 1\leq m\leq M.
\end{equation}
Here, $g_m(t)$ is a filtered function with the input $f(t)$ and the impulse response $h_m(t)$, and  $y_1(t)$, $y_2(t)$, $\cdots$, $y_M(t)$ are determined by   $h_1(t)$, $h_2(t)$, $\cdots$, $h_M(t)$.  The continuous function $\mathcal{T}_\mathbf{N} f(t) $ is called a multichannel interpolation (MCI) for $f$. The MCI reveals that one can reconstruct  a time-limited function $f$ by using  multiple types of samples simultaneously. If $f$ is periodic bandlimited,   it can be perfectly recovered by (\ref{MCI}).

It is noted that to find a function satisfying the interpolation consistency (\ref{consistency}) is to solve a system of $N_s=LM$ equations. And  the matrix involved in  this inverse problem  may have  a large condition number if  the  sample sets $\{g_m(\frac{2\pi p}{L}), 0\leq p\leq L-1\}$,    $ 1\leq m\leq M$  have a high   degree of relevance. In spite of this, in \cite{cheng2020Multicnon,cheng2019FFTMCI}, the authors showed  that the large scale ($N_s$) inverse problem could be converted  to a simple inversion   problem  of small matrices ($M\times M$) by    partitioning the frequency band into small pieces. Moreover,
the closed-form  of the MCI formula as well as   the FFT-based   implementation algorithm  (see Algorithm \ref{alg1})  were provided.

The MCI   guarantees that a  signal
can be well reconstructed from its clean  multichannel samples, little
has been said about the case where the samples are noisy.
It is of great significance to examine the errors that arise in the signal reconstruction by (\ref{MCI}) in the presence of noise.
In this paper, we consider the   reconstruction problem under the situation that a signal $f(t)$ is sampled in    a multichannel way and the samples are corrupted by the additive noise, i.e., we will use the noisy samples
\begin{equation}\label{noisydata}
   s_{m,p}=g_m(\tfrac{2\pi p}{L}) + \epsilon_{m,p},\quad  0\leq p\leq L-1,\quad 1\leq m\leq M,
\end{equation}
  to reconstruct $f(t)$. Here, $\{\epsilon_{m,p}\}$ is an i.i.d. noise process
with $\mathbb{E}[ \epsilon_{m,p} ] = 0$, $\mathrm{Var}[\epsilon_{m,p}] =  \sigma^2_{\epsilon}$.

The interpolation of noisy data introduces the undesirable error  in the reconstructed signal. There is a need to estimate the error of the MCI for the observations defined by (\ref{noisydata}). An accurate error estimate of the MCI in the presence of noise  helps to design suitable reconstruction formulas from noisy multichannel samples.  Note that the MCI applies to various kinds of sampling  schemes, thus the error analysis  can also be used to analyze what kinds of sampling schemes have a good performance in signal reconstruction  in the noisy environment. In the current paper, we provide an error estimate for the MCI from noisy multichannel samples, and express the error as a function of the sampling rate as well as the parameters associated with sampling schemes. In addition, we will show  how  sampling rate and   sampling  schemes affect the reconstruction error caused by noise.

Based on the error estimate of the MCI in the noisy environment, we will provide a class of signal reconstruction methods by introducing some reasonable smoothing and regularized corrections to the MCI such that the reconstructed signal  could be robust to noise. In other words, the reconstruction should not be affected much by small changes in the data. Besides, we need to make sure that the reconstructed signal will be convergent to the original signal as the sampling rate tends to infinity.

If $f(t)$ is a periodic bandlimited signal, only the error caused by noise needs to be considered. Otherwise,   the aliasing error should   be  taken into account as well. It is noted that the smoothing and regularization operations will restrain high frequency in general. It follows that to reduce the noise error   by the methods based on  smoothing or regularization may increase the aliasing error. Thus it is necessary to make a trade-off between the noise error and the aliasing error such that the reconstructed signal can be convergent to $f$ in the non-bandlimited case as the sampling rate tends to infinity.

The objective of this paper is to study the aforementioned problems that arise in the signal reconstruction from noisy multichannel data. The main contributions  are  summarized as follows.
\begin{enumerate}
  \item    The error estimate of the signal reconstruction by the MCI from noisy samples is given.
  \item  We propose four methods, i.e.,   post-filtering,   pre-filtering, $l_1$ regularization and $l_2$ regularization, to reduce the error caused by noise in the multichannel reconstruction. The parameters of post-filtering and pre-filtering are optimal in the sense of the expectation of mean square error (EMSE).
  \item  The convergence property of post-filtering is verified  theoretically and experimentally.  The numerical simulations as well as some guidelines to use the proposed signal reconstruction methods are also provided.
\end{enumerate}

The rest of the paper is organized as follows. Section  \ref{S2}
briefly  reviews the multichannel interpolation (MCI) and its FFT-based fast algorithm. The error estimate for the MCI of noisy samples is provided.  In Section  \ref{S3}, the techniques of post-filtering, pre-filtering and regularized approximation are applied
to reconstruct $f$ from its noisy multichannel samples. The comparative experiments for the different methods are conducted in Section \ref{S4}.  Finally,  conclusion and discussion are drawn at the end of the paper.

\section{Error analysis of the MCI from noisy samples}\label{S2}

\subsection{The MCI and its fast implementation algorithm}
We begin by reviewing the MCI in more detail. Let $N_1,N_2\in\mathbb{Z}$,  and $I^{\mathbf{N}}=\{n: N_1\leq n \leq N_2\}$, we denote by  $B_{\mathbf{N}}$ the totality of the periodic bandlimited functions (trigonometric polynomials) with the following form:
\begin{equation*}
f(t)=\sum_{n\in I^{\mathbf{N}}}a(n)e^{\qi nt}  ,~~~I^{\mathbf{N}}=\{n: N_1\leq n \leq N_2\}.
\end{equation*}
The bandwidth of $f$ is defined by   the  cardinality of $I^{\mathbf{N}}$, denoted by $\mu(I^{\mathbf{N}})$.
The set $I^{\mathbf{N}}$ can be expressed  as
$ I^{\mathbf{N}}=\bigcup_{j=1}^M I_j$, where
\begin{equation*}
I_j=\{n: N_1+(j-1)L\leq n\leq N_1+jL-1\}.
\end{equation*}

 We use the Fourier coefficients of $h_m$ to define the $M\times M$ matrix $$\mathbf{H}_n =\left[ b_m(n+jL-L)\right]_{jm}.$$  Suppose that $\mathbf{H}_n$ is invertible for every  $n\in I_1$ and denote its inverse matrix as
\begin{equation*}
\mathbf{H}_n^{-1}= \begin{bmatrix}
q_{11} (n) & q_{12} (n)&\cdots &q_{1M}(n)\\
q_{21}(n) & q_{22}(n) &\cdots &q_{2M} (n)\\
\vdots&\vdots& ~ & \vdots \\
q_{M1} (n)& q_{M2}(n) &\cdots &q_{MM}(n)
\end{bmatrix}.
\end{equation*}
Then the interpolating function $y_m$ in (\ref{MCI}) is given by
\begin{equation*}
y_m(t)= \sum_{n\in I^{\mathbf{N}}}r_{m}(n)e^{\qi n t},   \quad 1\leq m \leq M,
\end{equation*}
where
\begin{equation*}
r_{m}(n)= \begin{cases}
q_{mj}(n+L-jL),&  \text{if}~n\in I_j, ~j=1,2,\cdots,M ,  \\
0 & \text{if}~ n\notin I^{\mathbf{N}}.
\end{cases}
\end{equation*}

It was shown in \cite{cheng2019FFTMCI} that if $f$ is not bandlimited, the aliasing error of the MCI is given by
\begin{equation*}
\sum_{n\notin I^{\mathbf{N}}} \abs{a(n)}^2+ \sum_{k\notin\{1,2,\dots,M\}}\sum_{n\in  I_k}\abs{a(n)}^2 \sum_{l=1}^{M}\abs{\sum_{m=1}^M r_m(n+(l-k)L)b_m(n)}^2.
\end{equation*}
Moreover, the MCI  can be  implemented by a FFT-based algorithm (see Algorithm \ref{alg1}) and the well-known FFT interpolation \cite{fraser1989interpolation} is a special case of the MCI.

\IncMargin{1em} 
\begin{algorithm}[htb]

    \SetAlgoNoLine 
    \SetKwInOut{Input}{\textbf{Input}}\SetKwInOut{Output}{\textbf{Output}} 

    \Input{
    \\
    1. The multichannel samples $\mathbf{G}=[\mathbf{g}_1,\mathbf{g}_2,\dots,\mathbf{g}_M]$, where $\mathbf{g}_m$ is a  vector  \\~~~  consisting   of $L$ samples of $g_m$\;\\
    2.  The location of lower bound for the frequency band: $N_1$ \;\\
     3.   The number of  function values of  $\mathcal{T}_\mathbf{N} f(t)$:  $N_o$.\ \\}
    \Output{
        \\
         1.  The vector $\mathbf{f}_o$ consisting   of $N_o$ function values of $\mathcal{T}_\mathbf{N} f(t)$. \medskip \\}
    \BlankLine

    Multiply $k$-th row of $\mathbf{G}$ by $e ^{   \frac{-2\pi\qi N_1 (k-1)}{L} }$, obtain ${\mathbf{G}_e}$ \;
    Take FFT of ${\mathbf{G}_e}$ (for each column), obtain $\widetilde{\mathbf{G}}$\;
    Compute $\mathbf{A}_{L\times M}$, where $\mathbf{A}(k,:)=\widetilde{\mathbf{G}}(k,:)\mathbf{H}_{N_1+k-1}^{-1}$ \;
    Flatten $\mathbf{A}$  w.r.t. column, obtain $a$ ($N_s$ length vector) \;
    Zero padding: add $N_o-N_s$ zeros at the end of $\mathbf{a}$, obtain $\mathbf{a}_z$.\;
    Compute IFFT for $\mathbf{a}_z$, obtain $ {\mathbf{f}_e}$\;
    Multiply $k$-th element of ${\mathbf{f}_e}$ by $\frac{N_o}{L} e^{  \frac{2\pi\qi N_1 (k-1)}{N_o} } $,  obtain $\mathbf{f}_o$.
    \caption{FFT-based algorithm for MCI with complexity of $\mathcal{O}(N_o\log N_o)$.}\label{alg1}
\end{algorithm}
\DecMargin{1em}

\subsection{The error estimate for the MCI of noisy samples}

Given the noisy data (\ref{noisydata}), we define
\begin{equation*}
  f_{\mathbf{N},\epsilon}(t):= \frac{1}{L}\sum_{m=1}^{M} \sum_{p=0}^{L-1}\left(g_m(\tfrac{2\pi p}{L})+\epsilon_{m,p}\right)y_m(t-\tfrac{2\pi p}{L}).
\end{equation*}
If $f\in B_{\mathbf{N}}$, then
\begin{align*}
&\mathbb{E} \left(\frac{1}{2\pi} \int_{0}^{2\pi} \abs{ f_{\mathbf{N},\epsilon}(t)-f(t)}^2 dt\right) \\
=  &   \mathbb{E}\left( \frac{1}{2\pi}\int_{0}^{2\pi} \abs{  \frac{1}{L}\sum_{m=1}^{M} \sum_{p=0}^{L-1} \epsilon_{m,p} y_m(t-\tfrac{2\pi p}{L})}^2 dt\right)  \\
   = & \frac{1}{L^2}\mathbb{E} \sum_{m=1}^{M} \sum_{p=0}^{L-1}\sum_{m'=1}^{M} \sum_{p'=0}^{L-1} \epsilon_{m,p} \epsilon_{m',p'}\frac{1}{2\pi}\int_{0}^{2\pi}   y_m(t-\tfrac{2\pi p}{L})  \overline{y_{m'}(t-\tfrac{2\pi p'}{L}}) dt\\
   =& \frac{1}{L^2} \sigma_\epsilon^2 \sum_{m=1}^{M} \sum_{p=0}^{L-1} \left(\frac{1}{2\pi}\int_{0}^{2\pi}   \abs{y_m(t-\tfrac{2\pi p}{L}) }^2 dt\right)\\
   =& \frac{\sigma_\epsilon^2 }{L} \sum_{m=1}^{M}   \norm{y_m}_2^2= \frac{\sigma_\epsilon^2 }{L} \sum_{m=1}^{M} \sum_{n\in I^{\mathbf{N}}}\abs{r_{m}(n)}^2.
\end{align*}

Suppose  that $X,Y$ are independent random variables with the same normal distribution
$\mathcal{N}(0,\sigma^2)$, it is easy to verify that $\mathrm{Var}(X^2)=2\sigma^4$, $\mathrm{Var}(XY) = \sigma^4$.
Let
\begin{equation*}
 z(m,m',p,p'):= \frac{1}{2\pi}\int_{0}^{2\pi}   y_m(t-\tfrac{2\pi p}{L})  \overline{y_{m'}(t-\tfrac{2\pi p'}{L}}) dt
\end{equation*}
From H{\"o}lder inequality, we have that
\begin{equation*}
  \abs{z(m,m',p,p')}^2\leq \norm{y_m}_2^2\norm{y_{m'}}_2^2.
\end{equation*}
It follows that
\begin{align*}
&\mathrm{Var} \left(\frac{1}{2\pi} \int_{0}^{2\pi} \abs{ f_{\mathbf{N},\epsilon}(t)-f(t)}^2 dt\right) \\
=  &  \mathrm{Var}\left( \frac{1}{2\pi}\int_{0}^{2\pi} \abs{  \frac{1}{L}\sum_{m=1}^{M} \sum_{p=0}^{L-1} \epsilon_{m,p} y_m(t-\tfrac{2\pi p}{L})}^2 dt\right)  \\
   \leq & \frac{1}{L^4} \mathrm{Var} \sum_{m=1}^{M} \sum_{p=0}^{L-1}\sum_{m'=1}^{M} \sum_{p'=0}^{L-1} \epsilon_{m,p} \epsilon_{m',p'}\abs{z(m,m',p,p')}\\
   =& \frac{1}{L^4} \sum_{m=1}^{M} \sum_{p=0}^{L-1}\sum_{m'=1}^{M} \sum_{p'=0}^{L-1} \abs{z(m,m',p,p')}^2 \mathrm{Var} \left( \epsilon_{m,p} \epsilon_{m',p'}\right)\\
   \leq  &  \frac{2 \sigma_\epsilon^4 }{L^4} \sum_{m=1}^{M} \sum_{p=0}^{L-1}\sum_{m'=1}^{M} \sum_{p'=0}^{L-1} \abs{z(m,m',p,p')}^2\\
   \leq  &  \frac{2 \sigma_\epsilon^4 }{L^2} \sum_{m=1}^{M}  \sum_{m'=1}^{M}   \norm{y_m}_2^2\norm{y_{m'}}_2^2\\
   =  &  \frac{2 \sigma_\epsilon^4 }{L^2} \left(\sum_{m=1}^{M} \norm{y_m}_2^2\right)^2
\end{align*}
Therefore the variance of   mean square error is bounded and is not larger than  twice the square of  the expectation.

In order to show the mean square error of MCI caused by noise  more clearly, we consider three concrete sampling schemes, namely, the reconstruction  problem of $f$ from (1)  the samples of $f$ (single-channel); (2) the  samples of $f$ and $\mathcal{H}f $ (two-channel); (3) the samples of $f$ and $f'$ (two-channel).  For simplicity, we abbreviate the MCI of the above types of samples   as F1, FH2 and FD2 respectively and denote by $N_s=LM$ the total  number of samples. For F1, we have that $M=1$, $N_s = LM=L$. It easy to see that
\begin{equation*}
  r(n,\mathrm{F1},N_s) =1 ~~\text{for}~~-\frac{N_s}{2}+1\leq n \leq \frac{N_s}{2}.
\end{equation*}
For FH2, we have that $M=2$, $N_s =2L$. Since
\begin{equation*}
	 \mathbf{H}_n= \begin{bmatrix}
	1 & - \qi \sgn (n)   \\
	1 & - \qi \sgn(n+L)
	\end{bmatrix}.
	\end{equation*}
	It is clear that
	\begin{equation*}
	 \mathbf{H}_n^{-1}= \begin{bmatrix}
	\frac{1}{2} &\frac{1}{2}  \\
	- \frac{\qi}{2} & \frac{\qi}{2}
	\end{bmatrix} ~~\text{for}~~-L+1\leq n \leq -1,~~ \mathbf{H}_0^{-1}= \begin{bmatrix}
	1 & 0 \\
	- \qi  &1
	\end{bmatrix}.
\end{equation*}
It follows that
\begin{equation*}
r_{1}(n,\mathrm{FH2},N_s)= \begin{cases}
 \frac{1}{2},&  \text{if}~~ 1\leq \abs{n} \leq L-1 ,    \\
0 & \text{if}~~n=L,\\
1 & \text{if}~~n=0.
\end{cases}
\end{equation*}
\begin{equation*}
r_{2}(n,\mathrm{FH2},N_s)= \begin{cases}
 -\frac{ \qi}{2},&  \text{if}~~ -L+1\leq n \leq -1,    \\
\frac{ \qi}{2}  & \text{if}~~ 1\leq n \leq L-1,\\
-\qi & \text{if}~~n=0,\\
1& \text{if}~~n=L.
\end{cases}
\end{equation*}
For FD2, by direct computations, we have that
\begin{equation*}
	 \mathbf{H}_n= \begin{bmatrix}
 1 & \qi n \\
 1 & \qi (L+n)
	\end{bmatrix},\quad
	 \mathbf{H}_n^{-1}= \begin{bmatrix}
		 \frac{L+n}{L} & -\frac{n}{L} \\
 \frac{\qi}{L} & -\frac{\qi}{L}
	\end{bmatrix} .
\end{equation*}
It follows that
\begin{equation*}
r_{1}(n,\mathrm{FD2},N_s)= \begin{cases}
 1+\frac{n}{L},&  \text{if}~~ -L+1\leq n \leq 0,      \\
1 -\frac{n}{L} & \text{if}~~  1\leq n \leq L.
\end{cases}
\end{equation*}
\begin{equation*}
r_{2}(n,\mathrm{FD2},N_s)= \begin{cases}
 \frac{ \qi}{L},&  \text{if}~~ -L+1\leq n \leq 0,    \\
-\frac{ \qi}{L}  & \text{if}~~ 1\leq n \leq L.
\end{cases}
\end{equation*}

To study FH2 and FD2, we assume that $N_s$ is an even number and $I^{\mathbf{N}}  = \{n: -\frac{N_s}{2}+1\leq n \leq \frac{N_s}{2}\}$. It should be noted that $L= {N_s} $ for F1 because it is a single-channel interpolation. In contrast, $L=\frac{N_s}{2}$  for FH2 and FD2 as they are  two-channel interpolations. Thus, to compare the performance of the three interpolation methods  under the same total number of samples $N_s$, one needs to keep in mind that $L$ has different values for F1 and FH2.

\begin{figure}
  \centering
  \includegraphics[width=15.9cm]{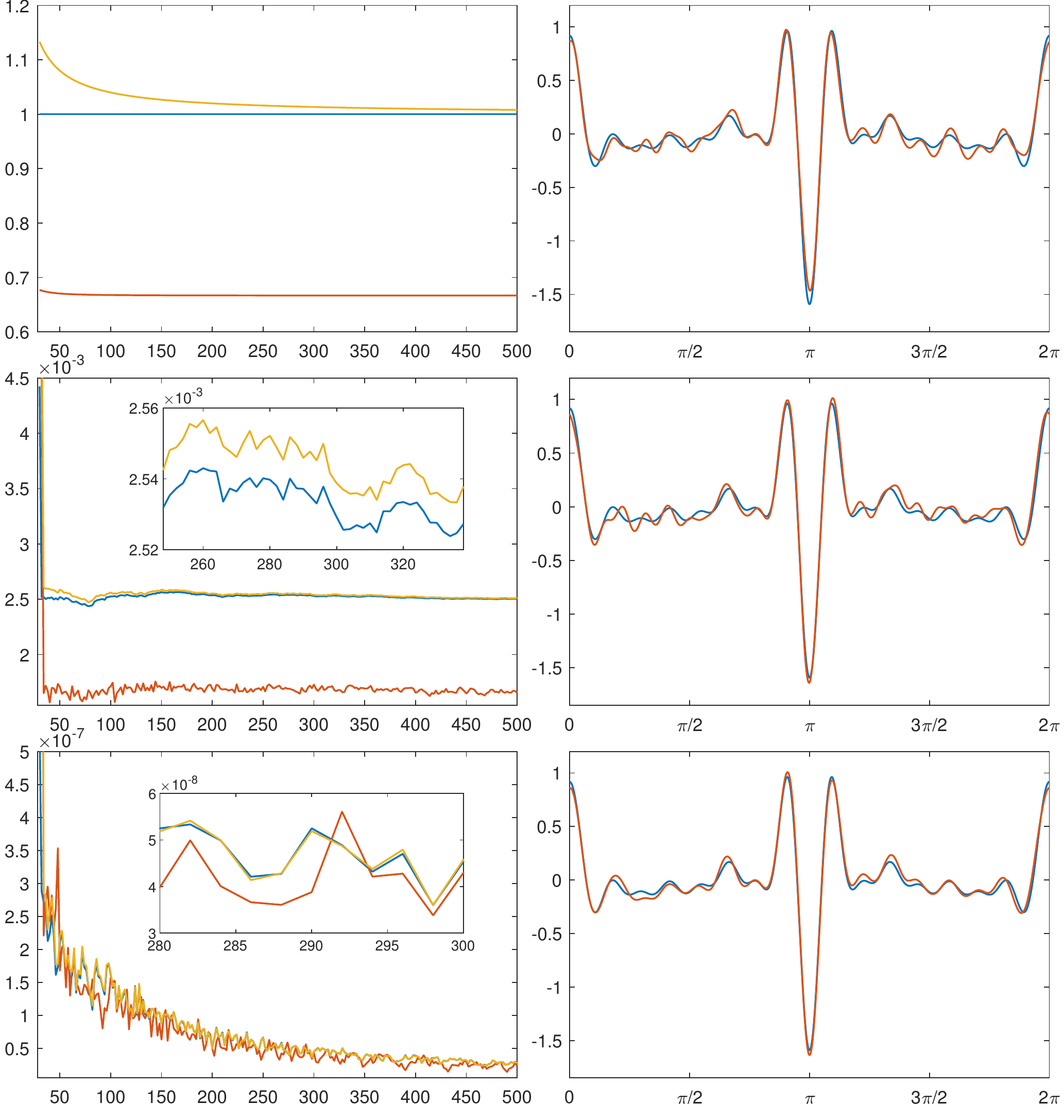}
  \caption{Left 1: $\frac{1 }{L} \sum_{m=1}^{M} \sum_{n\in I^{\mathbf{N}}}\abs{r_{m}(n)}^2$ for F1 (blue), FH2 (yellow) and FD2 (orange). Left 2: the average of the mean square error of $50$ experiments for F1 (blue), FH2 (yellow) and FD2 (orange). Left 3: the variance of the mean square error of $50$ experiments for F1 (blue), FH2 (yellow) and FD2 (orange). The horizontal axis in Left 1-3 represents  the total number of samples, namely $N_s$. Right 1-3 are the reconstructed results (red lines) for F1, FH2 and FD2 respectively in the noisy environment with $N_s=56, \sigma_{\epsilon}=0.05$; the blue line is the  original test function $\phi_B(t)$. }\label{noise-error}
\end{figure}

Having  introduced the Fourier coefficients of the interpolation functions for F1, FH2 and FD2,  we have that
\begin{equation*}
   \frac{1}{N_s}\sum_{n\in I^{\mathbf{N}}}\abs{r(n,\mathrm{F1},N_s)}^2  =1,
\end{equation*}
\begin{equation*}
  \frac{1 }{L} \sum_{m=1}^{M} \sum_{n\in I^{\mathbf{N}}}\abs{r_{m} (n,\mathrm{FH2},N_s )}^2 =1+\frac{4}{N_s},
\end{equation*}
\begin{equation*}
  \frac{1 }{L} \sum_{m=1}^{M} \sum_{n\in I^{\mathbf{N}}}\abs{r_{m} (n,\mathrm{FD2},N_s )}^2 = \frac{2}{3}+\frac{28}{3  {N_s}^2}.
\end{equation*}
 Besides the theoretical error estimate, the experiments are conducted  to compare the reconstructed results by  F1, FH2 and FD2. Let
\begin{equation}\label{testfunc}
\phi(z) = \frac{0.08 z^2+0.06z^{10}}{(1.3-z)(1.5-z)} + \frac{0.05z^3+0.09z^{10}}{(1.2+z)(1.3+z)},
\end{equation}
 \begin{equation}\label{Dkernel}
D(t,k_1,k_2) =  \sum_{n=k_1}^{k_2}e^{\qi n t},
\end{equation}
  \begin{equation*}
\phi_B(t)=   \phi(e^{\qi   t})*D(t,-16,16).
\end{equation*}
If $k_2=-k_1>0$,  $D(t,k_1,k_2)$  is the Dirichlet kernel of order $k_2$.  We use $\phi_B(t)$ as the test function. Obviously, it is bandlimited  with the bandwidth $33$.
 The theoretical errors, the experimental errors    and  the reconstructed results are shown in Figure \ref{noise-error} and some conclusions can be drawn   as follows.
\begin{enumerate}
  \item  FD2 performs better than F1 in terms of noise immunity and FH2 has the worst performance.
  \item  As the total  number of samples increases, the expectation of   mean square error (EMSE) would not  decrease  if there is no  additional correction   made in the multichannel reconstruction.
  \item   The variance of   mean square error (VMSE) is bounded and it  decreases as $N_s$  increases.
\end{enumerate}
\begin{remark}
   In the second  row first column of Figure \ref{noise-error},  we see that  the errors of F1, FH2 and FD2  become   significantly large  when $N_s<33$. This is because  the test function has the bandwidth $33$, the  reconstruction error is caused   not only   by noise but also by aliasing.
\end{remark}

\section{Multichannel reconstruction from noisy samples}\label{S3}

The  MCI cannot work well if one observes noisy data because $ f_{\mathbf{N},\epsilon}(t) $ does not converge to $f(t)$ in the sense of the expectation of  mean square error (EMSE).  To alleviate this problem, some smoothing corrections are required. If $f(t)$ is bandlimited and the number of samples  is larger than the bandwidth, we only need to consider the error caused by noise.  Suppose that $f\in B_\mathbf{K}, I^{\mathbf{K}}=\{k,K_1\leq k\leq K_2\}$ and $\mu(I^{\mathbf{K}})\leq \mu(I^{\mathbf{N}})=N_s$, where $N_s$ is the total number of samples. In this section, the techniques of post-filtering, pre-filtering,  regularized approximation   are applied to reconstruct $f$ from noisy  samples.

\subsection{Post-filtering}

In \cite{Pawlak2003postfilter}, the ideal low-pass post-filtering is  applied to the Shannon sampling formula and the error of signal reconstruction is also evaluated. Different from the previous work, we first derive the EMSE of  the reconstruction by   MCI  and post-filtering. Then the filter is obtained by solving the optimization problem that minimizes the EMSE.

\subsubsection{Formulation of post-filtering}

 A natural smoothing approach for the reconstructed signal is to convolute $ f_{\mathbf{N},\epsilon}(t) $ with a function $w(t)   \in B_\mathbf{K}$. Let
\begin{equation*}
  \widetilde{f}(t,N_s,\mathbf{K}) =  (f_{\mathbf{N},\epsilon}*w)(t).
\end{equation*}
 Note that
\begin{equation*}
   f*D(\cdot,K_1,K_2)(t) =f(t)
\end{equation*}
provided that $f\in B_\mathbf{K}$. It follows that
 \begin{align*}
 & \widetilde{f}(t,N_s,\mathbf{K})-f(t) \\
    =&   f_{\mathbf{N},\epsilon}*w (t)- f*D(\cdot,K_1,K_2)(t)\\
    =  & \left[  f*(w-D(\cdot,K_1,K_2)) \right] (t) + \frac{1}{L}\sum_{m=1}^{M} \sum_{p=0}^{L-1}\epsilon_{m,p} \left[y_m*w\right](t-\tfrac{2\pi p}{L}).
 \end{align*}
Since $\{\epsilon_{m,p}\}$ is an i.i.d. noise process
with $\mathbb{E}[ \epsilon_{m,p} ] = 0$, then
\begin{align*}
    & \mathbb{E} \left(\abs{\widetilde{f}(t,N_s,\mathbf{K})-f(t)}^2\right) \\
  = &\abs{\left[  f*(w-D(\cdot,K_1,K_2)) \right] (t)}^2+\frac{1}{L^2}\sum_{m=1}^{M} \sum_{p=0}^{L-1}\abs{\left[y_m*w\right](t-\tfrac{2\pi p}{L})}^2    \mathbb{E} [\epsilon_{m,p}^2]\\
  =&\abs{\left[  f*(w-D(\cdot,K_1,K_2)) \right] (t)}^2+ \frac{1}{L^2}\sum_{m=1}^{M} \sum_{p=0}^{L-1}\abs{\left[y_m*w\right](t-\tfrac{2\pi p}{L})}^2   \sigma_{\epsilon}^2.
\end{align*}
Denote  the Fourier coefficient of $w$ by   $\beta_k$, it follows that
\begin{align*}
&\mathbb{E} \left(\frac{1}{2\pi} \int_{0}^{2\pi} \abs{ \widetilde{f}(t,N_s,\mathbf{K})-f(t)}^2 dt\right) \\
   = &  \frac{1}{2\pi}\int_{0}^{2\pi}  \mathbb{E}\abs{\widetilde{f}(t,N_s,\mathbf{K})-f(t)}^2 dt\\
   =&\norm{   f*(w-D(\cdot,K_1,K_2))   }_2^2+ \frac{\sigma_\epsilon^2 }{L } \sum_{m=1}^{M}   \norm{y_m *w}_2^2\\
      =& \sum_{k=K_1}^{K_2} \abs{a(k) (\beta_k-1)}^2 + \frac{\sigma_\epsilon^2 }{L } \sum_{m=1}^{M}  \sum_{k=K_1}^{K_2} \abs{r_m(k,\mathrm{Type},N_s)\beta_k}^2 .
\end{align*}
 \begin{remark}
    Since the functions considered here are square integrable, the interchange of expectation and integral  is   permissible by the dominated convergence theorem. There are some similar cases happening elsewhere in the paper, we will omit the explanations.
 \end{remark}

Let $\bm{\beta}  = (\beta_{K_1} ,\cdots,\beta_{K_2} )^\text{T}$ and
\begin{equation}\label{opti_pro_post}
\Phi_1(\bm{\beta})= \sum_{k=K_1}^{K_2} \abs{a(k) (\beta_k-1)}^2 + \frac{\sigma_\epsilon^2 }{L } \sum_{m=1}^{M}  \sum_{k=K_1}^{K_2} \abs{r_m(k,\mathrm{Type},N_s)\beta_k}^2 .
\end{equation}
Since $\abs{\abs{\beta_k}-1}\leq \abs{\beta_k-1}$ and the equality holds only if
$\beta_k\geq 0$, it follows that
\begin{equation*}
\Phi_1(\bm{\beta}_+) -\Phi_1(\bm{\beta}) = \sum_{k=K_1}^{K_2}\abs{a(k)} \left(\abs{\abs{\beta_k}-1}-\abs{\beta_k-1}\right)\leq 0,
\end{equation*}
where $\bm{\beta}_+ =(\abs{\beta_{K_1}} ,\cdots,\abs{\beta_{K_2} })^\text{T} $. Thus, if
\begin{equation*}\label{key}
 \bm{\beta}^* = (\beta_{K_1}^*,\cdots,\beta_{K_2}^*)^\text{T} = \mathop{\arg\min}\limits_{\bm{\beta}} \Phi_1(\bm{\beta}),
\end{equation*}
 then $\beta_{k}^*\geq 0$ for every $K_1\leq k \leq K_2$. To minimize $ \Phi_1(\bm{\beta})$, we rewrite it as follows:
\begin{equation*}
\Phi_1(\bm{\beta}) = \norm{\mathbf{A}_+\bm{\beta}-\mathbf{a}_+}_2^2 + \frac{\sigma_\epsilon^2 }{L }\sum_{m=1}^{M} \norm{\mathbf{R}_{m,+}\bm{\beta}}_2^2,
\end{equation*}
where
\begin{equation*}
\mathbf{a}_+ = \left(\abs{a(K_1)} ,\cdots,\abs{a(K_2)}\right)^\text{T},\quad \mathbf{A}_+ = \operatorname{diag}(\mathbf{a}_+),
\end{equation*}
\begin{equation*}
\mathbf{R}_{m,+} = \operatorname{diag}(\abs{r_m(K_1,\mathrm{Type},N_s)},\cdots,\abs{r_m(K_2,\mathrm{Type},N_s)}).
\end{equation*}
Differentiating $\Phi_1(\bm{\beta})$ with respect to $\bm{\beta}$ and solving $\nabla \Phi_1(\bm{\beta})=0$, we obtain the optimal solution for minimizing the expectation of   mean square error. That is,
\begin{equation}\label{opt_post_filter}
 \bm{\beta}^* = \left(\mathbf{A}_+ ^\text{T}\mathbf{A}_+  +\frac{\sigma_\epsilon^2 }{L }\sum_{m=1}^{M} \mathbf{R}_{m,+}^\text{T} \mathbf{R}_{m,+} \right)^{-1} \mathbf{A}_+ ^\text{T}\mathbf{a}_+.
\end{equation}

\subsubsection{Estimation of spectral density}\label{s312}

The formula (\ref{opt_post_filter}) gives the optimal values for the parameters of   post-filtering, to minimize the difference (EMSE) between the filtered and the original (clean)  signal $f(t)$. The key problem is that the square of absolute value of $a(n)$, namely the spectral density of $f(t)$, is unknown in typical cases.  Thus we have to estimate the value of $\abs{a(n)}^2$ from the noisy multichannel samples.

There are various techniques for spectral  density estimation. The representative methods are periodogram, Welch's method, autoregressive model  and moving-average model,   etc.  Here, we provide an unbiased estimation for $\abs{a(n)}^2$ by using the uncorrelatedness  of signal and noise.

Let
\begin{equation*}
\mathbf{s}_m = (s_{m,0},s_{m,1},\cdots,s_{m,L-1})^{\text{T}},\quad 1\leq m \leq M,
\end{equation*}
\begin{equation*}
\mathbf{g}_m = (g_m(t_0),g_m(t_1),\cdots, g_m(t_{L-1}))^{\text{T}},\quad t_p=\tfrac{2\pi p}{L},1\leq m \leq M,
\end{equation*}
\begin{equation*}
\bm{\epsilon}_m = (\epsilon_{m,0},\epsilon_{m,1},\cdots,\epsilon_{m,L-1})^{\text{T}},\quad 1\leq m \leq M,
\end{equation*}
then
\begin{equation*}
  \mathbf{s}_m = \mathbf{g}_m + \bm{\epsilon}_m .
\end{equation*}
To estimate  $\abs{a(n)}^2$, we need to introduce the vector $\mathbf{d}_0$ and $\mathbf{d}_\epsilon$, where
\begin{equation}\label{ddd}
   \mathbf{d}_0=\frac{1}{L} \begin{bmatrix}
	\F_L\U_L\mathbf{g}_1 \\
	\F_L\U_L\mathbf{g}_2 \\
	\vdots \\
	\F_L\U_L\mathbf{g}_M \\
	\end{bmatrix},\quad   \mathbf{d}_\epsilon=\frac{1}{L} \begin{bmatrix}
	\F_L\U_L\mathbf{s}_1 \\
	\F_L\U_L\mathbf{s}_2 \\
	\vdots \\
	\F_L\U_L\mathbf{s}_M \\
	\end{bmatrix}.
\end{equation}
Here, $\F_L$ is  the $L$-th order DFT matrix
	\begin{equation}\label{DFTmatrix}
	\F_L=\begin{bmatrix}
	\omega^0 & \omega^0& \omega ^0&\cdots &\omega^0\\
	\omega^0 & \omega^1& \omega^2 &\cdots &\omega^{L-1}\\
	\omega^0 & \omega^2& \omega^4 &\cdots &\omega^{2(L-1)}\\
	\vdots&\vdots&\vdots& \ddots & \vdots \\
	\omega^0 & \omega^{L-1}& \omega^{2(L-1)} &\cdots &\omega^{(L-1)^2}\\
	\end{bmatrix}
	\end{equation}
	with $\omega=e^{ {-2\pi\qi}/{L}}$ and $\U_L$ is a diagonal matrix
	\begin{equation}\label{shiftmatrix}
	\U_L= \begin{bmatrix}
	\omega^0\\
	&\omega^{N_1}& &\text{{\huge 0}}\\
	& & \omega^{2{N_1}} \\
	& \text{{\huge0}} & & \ddots\\
	& & & & \omega^{(L-1){N_1}}
	\end{bmatrix}.
	\end{equation}

Note that $\F_L^{*}=L \F_L^{-1}$, $\U_L^{*}=\U_L^{-1}$ and $\{\epsilon_{m,p}\}$ is an i.i.d. noise process, it follows that
\begin{equation*}
  \mathbb{E} [  \mathbf{d}_\epsilon \mathbf{d}_\epsilon^*] =\frac{\sigma_\epsilon^2}{L} \mathbf{I} +\mathbf{d}_0 \mathbf{d}_0^*.
\end{equation*}
Let $\mathbf{B}$ be a $N_s$ by $N_s$ matrix and the entry in the $m$-th row and $n$-th column of  $\mathbf{B}$ is
\begin{equation*}
\mathbf{B}(m,n)= \begin{cases}
\mathbf{H}_{N_1+k-1}^{-1}(j+1,i+1),&  \text{if}~~ m=iL+k , n=jL+k,   \\
0 & \text{otherwise},
\end{cases}
\end{equation*}
where $1\leq k\leq L$ and  $0\leq i,j \leq M-1$.
By direct computations, we have that
\begin{equation*}
\mathbb{E} [ \mathbf{B} \mathbf{d}_\epsilon \mathbf{d}_\epsilon^*\mathbf{B}^*]= \mathbf{B} (\frac{\sigma_\epsilon^2}{L} \mathbf{I} +\mathbf{d}_0\mathbf{d}_0^*) \mathbf{B}^* = \frac{\sigma_\epsilon^2}{L}\mathbf{B}\mathbf{B}^*+\mathbf{B} \mathbf{d}_0 \mathbf{d}_0^*\mathbf{B}^*.
\end{equation*}
If $f$ is bandlimited, it can be verified that the diagonal element of $\mathbf{B} \mathbf{d}_0 \mathbf{d}_0^*\mathbf{B}^*$ is equal to  $\abs{a(n)}^2$ (by a similar method for proving Lemma 1 in \cite{cheng2019FFTMCI}). It follows that the diagonal element of
\begin{equation}\label{spectral_estimate}
 \mathbf{B} \mathbf{d}_\epsilon \mathbf{d}_\epsilon^*\mathbf{B}^* -   \frac{\sigma_\epsilon^2}{L}\mathbf{B}\mathbf{B}^*
\end{equation}
is an unbiased estimation for $\abs{a(n)}^2$.

To  validate the effectiveness of the above method for estimating   spectral density, the noisy multichannel samples are applied to estimate  $\abs{a(n)}^2$ by the formula (\ref{spectral_estimate}) experimentally. We will perform a series of experiments   under different   quantities and types of samples. Let
\begin{equation}\label{testfunc1}
f(t) = \sum_{n=N_1}^{N_2} a(n)e^{\qi nt},\quad N_1=-2,N_2=3
\end{equation}
be the test function, where $ a(-2)=1+\qi,a(-1)=2-\qi,a(0)=1,a(1)=2+\qi,a(2)=1-\qi,a(3)=0 $. The  mean square error (MSE)  for estimating the spectral density of $f$ is
defined by
\begin{equation*}
\delta_{sde} = \frac{\sum_{n=N_1}^{N_2}\abs{\abs{a(n)}^2-\tilde{A}(n)}^2}{N_2-N_1+1},
\end{equation*}
where $\tilde{A}(n)$ is the $(N_1-n+1)$-th diagonal element of $\mathbf{B} \mathbf{d}_\epsilon \mathbf{d}_\epsilon^*\mathbf{B}^* -   \frac{\sigma_\epsilon^2}{L}\mathbf{B}\mathbf{B}^*$.
 To show the performance of the estimation more accurately, each  experiment will be repeated $1000$ times and the corresponding average MSE  is an approximation of the expectation of MSE.

\begin{table}[ht]
\centering
\caption{The experimental results of multichannel based method for spectral density estimation. The first row displays the total number of samples used in each experiment. The second and third rows display the number of samples of $f$ and $f'$ used in each experiment respectively. The error of  spectral density estimation is given in the last row.} \vspace{0.2cm}
\begin{tabular}{|c|c|c|c|c|c|c|c|c|}
	\hline
  $N_s$ & $6$ & $6$ &  $30$& $60$ & $60$ &$300$  & $600$ &$600$  \\
	\hline
  $f$	& $6$ & $3$ & $30$ & $60$ & $30$  & $300$  & $600$ &$300$  \\
	\hline
  $f'$	& $0$ & $3$ & $0$ & $0$ & $30$  & $0$  & $0$ & $300$ \\
	\hline
Average MSE	& $0.3390$ &  $0.3427$& $0.0673$ & $0.0335$ & $0.0356$ & $0.0071$ & $0.0034$ & $0.0036$ \\
	\hline
\end{tabular}\label{table1}
\end{table}

The experimental results are presented in Table \ref{table1}.  The second column indicates that if we use $6$ samples of $f$ to estimate   spectral density, the expectation of MSE is   approximately equal to $0.3390$. It can be seen that the expectation of MSE for spectral density estimation varies in inverse proportion to the total number of samples. In other words, the experimentally obtained MSE, i.e. $\delta_{sde}$, tends to $0$ as the total number of samples goes to infinity and  if the same total  number of samples are used to estimate  spectral density, the fluctuations of  MSE
caused by different sampling schemes are not significant. Besides, it is noted that the traditional single-channel  based  method for   spectral density estimation    can not utilize the multichannel information to improve the accuracy. By contrast, the proposed multichannel based method  fuses the different types of samples, thereby extending the scope of application  and enhancing the precision, as seen from the column four and six of Table  \ref{table1}.

\subsection{Pre-filtering}

If $f$ is bandlimited, it can be expressed as
\begin{equation*}
  f(t)= \frac{1}{L} \sum_{m=1}^{M} \mathbf{g}_m^{\text{T}} \U_L \F_L  \mathbf{v}_m(t),
\end{equation*}
where
	\begin{equation*}
	 {\mathbf{v}}_m(t)=\left(v_{m,N_1}(t), v_{m,N_1+1}(t),\cdots, v_{m,L+N_1-1}(t)\right)^{\text{T}},\quad v_{m,n}(t)=\sum_{k=1}^M  q_{mk}(n) e^{\qi (n+kL-L)t}
	\end{equation*}
	  for  $ n\in I_1$. We consider to filter the noisy multichannel samples $\mathbf{s}_1,\mathbf{s}_2,\cdots,\mathbf{s}_M$ by modifying the  frequency components in the DFT domain. Let
\begin{equation*}
  \tilde{\mathbf{s}}_m = \U_L^{-1} \F_L^{-1} \bm{\Lambda}_m \F_L\U_L   \mathbf{s}_m, \quad  \bm{\Lambda}_m = \operatorname{diag}(\lambda_{m, N_1},\lambda_{m, N_1+1},\cdots,\lambda_{m,L+N_1-1}),
\end{equation*}
and construct a function of form
\begin{equation*}
  \breve{f}(t) = \frac{1}{L} \sum_{m=1}^{M} {\tilde{\mathbf{s}}_m}^{\text{T}} \U_L \F_L  \mathbf{v}_m(t) =  \frac{1}{L} \sum_{m=1}^{M}  \mathbf{s}_m ^{\text{T}}\U_L \F_L \bm{\Lambda}_m\mathbf{v}_m(t) .
\end{equation*}

In this part, we want to determine the values of $\lambda_{m,n}$, $1\leq m\leq M,n\in I_1$, such that  $\breve{f}(t)$ be a good estimation of $f(t)$. The square of absolute value for the difference of $f(t)$ and $\breve{f}(t)$ is
 \begin{align*}
   \abs{ \breve{f}(t) -f(t)}^2 & =\abs{\frac{1}{L} \sum_{m=1}^{M}  (\mathbf{g}_m ^{\text{T}}+\bm{\epsilon}_m^{\text{T}} )\U_L \F_L \bm{\Lambda}_m\mathbf{v}_m(t)-\frac{1}{L} \sum_{m=1}^{M} \mathbf{g}_m^{\text{T}} \U_L \F_L  \mathbf{v}_m(t) }^2 \\
     & = \abs{\frac{1}{L} \sum_{m=1}^{M}   \mathbf{g}_m ^{\text{T}} \U_L \F_L (\bm{\Lambda}_m-\mathbf{I})\mathbf{v}_m(t)+\frac{1}{L} \sum_{m=1}^{M} \bm{\epsilon}_m^{\text{T}} \U_L \F_L  \mathbf{v}_m(t) }^2.
 \end{align*}
 \begin{remark}
 To blend the information of   different types of samples, the frequency band has to be partitioned (shown as follows for the case of $M=3,L=5$).   The vector $\mathbf{v}_m(t)$ plays an important role in the multichannel reconstruction. The frequency bands of the  $L$ elements for $\mathbf{v}_m(t)$ are located in the $L$ positions corresponding to different colors. The role of $\U_L$ is to shift the zero-frequency component to the center of spectrum.
   \begin{center}
\begin{tikzpicture}
\draw[pink] (0,0) rectangle (0.5,0.5);
\fill[pink] (0,0) rectangle (0.5,0.5);
\draw[blue] (0.5,0) rectangle (1,0.5);
\fill[blue] (0.5,0) rectangle (1,0.5);
\draw[red] (1,0) rectangle (1.5,0.5);
\fill[red] (1,0) rectangle (1.5,0.5);
\draw[cyan] (1.5,0) rectangle (2,0.5);
\fill[cyan] (1.5,0) rectangle (2,0.5);
\draw[teal] (2,0) rectangle (2.5,0.5);
\fill[teal] (2,0) rectangle (2.5,0.5);

\draw[pink] (2.5,0) rectangle (3,0.5);
\fill[pink] (2.5,0) rectangle (3,0.5);
\draw[blue] (3,0) rectangle (3.5,0.5);
\fill[blue] (3,0) rectangle (3.5,0.5);
\draw[red] (3.5,0) rectangle (4,0.5);
\fill[red] (3.5,0) rectangle (4,0.5);
\draw[cyan] (4,0) rectangle (4.5,0.5);
\fill[cyan] (4,0) rectangle (4.5,0.5);
\draw[teal] (4.5,0) rectangle (5,0.5);
\fill[teal] (4.5,0) rectangle (5,0.5);

\draw[pink] (5,0) rectangle (5.5,0.5);
\fill[pink] (5,0) rectangle (5.5,0.5);
\draw[blue] (5.5,0) rectangle (6,0.5);
\fill[blue] (5.5,0) rectangle (6,0.5);
\draw[red] (6,0) rectangle (6.5,0.5);
\fill[red] (6,0) rectangle (6.5,0.5);
\draw[cyan] (6.5,0) rectangle (7,0.5);
\fill[cyan] (6.5,0) rectangle (7,0.5);
\draw[teal] (7,0) rectangle (7.5,0.5);
\fill[teal] (7,0) rectangle (7.5,0.5);

\draw[densely dashed] (0,-0.25) rectangle (7.5,0.75);
\draw[thick] (2.5,-0.25) to (2.5,0.75);
\draw[thick] (5,-0.25) to (5,0.75);
\end{tikzpicture}
\end{center}
\end{remark}

Let
\begin{align*}
  u_1(t) & =\frac{1}{L} \sum_{m=1}^{M}   \mathbf{g}_m ^{\text{T}} \U_L \F_L (\bm{\Lambda}_m-\mathbf{I})\mathbf{v}_m(t), \\
  u_2(t) & =\frac{1}{L} \sum_{m=1}^{M} \bm{\epsilon}_m^{\text{T}} \U_L \F_L \bm{\Lambda}_m \mathbf{v}_m(t).
\end{align*}
We have that
\begin{align*}
  \mathbb{E}\left(\abs{u_2(t)}^2\right)   & =\mathbb{E} \left(\frac{1}{L^2}\sum_{m=1}^{M}\abs{\bm{\epsilon}_m^{\text{T}} \U_L \F_L \bm{\Lambda}_m \mathbf{v}_m(t)}^ 2 \right) \\
    &  =\mathbb{E} \left(\frac{1}{L^2}\sum_{m=1}^{M}\abs{\bm{\epsilon}_m^{\text{T}}  \mathbf{z}_m(t)}^ 2 \right)\\
    & = \mathbb{E} \left(\frac{1}{L^2}\sum_{m=1}^{M}\sum_{p\in I_1}\sum_{k\in I_1} \epsilon_{m,p-N_1+1}\overline{\epsilon_{m,k-N_1+1}}z_{m,p}(t)\overline{z_{m,k}(t)}\right)\\
    &= \frac{\sigma_{\epsilon}^2}{L^2}\sum_{m=1}^{M}\sum_{p\in I_1}\abs{z_{m,p}(t)}^2.
\end{align*}
Here, we denote $\U_L \F_L \bm{\Lambda}_m \mathbf{v}_m(t)$ by
\begin{equation*}
  \mathbf{z}_m(t) = \left(z_{m,N_1}(t), z_{m,N_1+1}(t),\cdots, z_{m,L+N_1-1}(t)\right)^{\text{T}}
\end{equation*}
 and the first equality and last equality  are  direct consequences of the independence of noise. By the definition of $r_m(n)$, we know that $q_{mk}(n) =r_m(n+kL-L)$ for $n\in I_1 $, it follows that
 \begin{equation*}
   v_{m,n}(t)=\sum_{k=1}^{M} r_m(n+kL-L) e^{\qi (n+kL-L)t},\quad n\in  I_1.
 \end{equation*}
Thus
\begin{align*}
  z_{m,p}(t) & = \omega^{(p-N_1) N_1} \sum_{n\in I_1} \lambda_{m,n} v_{m,n}(t) \omega^{(n-N_1)(p-N_1)} \\
    &  = \sum_{k=1}^{M} \sum_{n\in I_1}  \omega^{ n (p-N_1)}\lambda_{m,n}r_m(n+kL-L) e^{\qi (n+kL-L)t}.
\end{align*}
It follows from the Parseval's identity  that
\begin{align*}
  \mathbb{E}\left(\frac{1}{2\pi} \int_{0}^{2\pi}\abs{u_2(t)}^2 dt\right) &  =  \frac{1}{2\pi} \int_{0}^{2\pi} \mathbb{E}\left(\abs{u_2(t)}^2\right) dt \\
   & = \frac{\sigma_{\epsilon}^2}{L^2}\sum_{m=1}^{M}\sum_{p\in I_1}   \frac{1}{2\pi} \int_{0}^{2\pi} \abs{z_{m,p}(t)}^2dt \\
   & = \frac{\sigma_{\epsilon}^2}{L^2}\sum_{m=1}^{M}\sum_{p\in I_1}\sum_{k=1}^{M} \sum_{n\in I_1} \abs{\lambda_{m,n}r_m(n+kL-L)}^2\\
   & = \frac{\sigma_{\epsilon}^2}{L }\sum_{m=1}^{M} \sum_{k=1}^{M} \sum_{n\in I_1} \abs{\lambda_{m,n}r_m(n+kL-L)}^2.
\end{align*}

By direct computations, we have that
\begin{align*}
    u_1(t) &= \sum_{m=1}^{M}\sum_{n\in I_1} (\lambda_{m,n}-1) d_{m}(n)v_{m,n}(t),\\
    &  =  \sum_{m=1}^{M}\sum_{n\in I_1} \sum_{k=1}^{M}(\lambda_{m,n}-1) d_{m}(n) r_m(n+kL-L) e^{\qi (n+kL-L)t}
\end{align*}
where
\begin{equation*}
  d_m(n) = \sum_{k=1}^{M}a(n+kL-L)b_m(n+kL-L).
\end{equation*}
It follows from the Parseval's identity  that
\begin{equation*}
   \frac{1}{2\pi} \int_{0}^{2\pi}  \abs{u_1(t)}^2  dt = \sum_{k=1}^{M}\sum_{n\in I_1}\abs{ \sum_{m=1}^{M}(\lambda_{m,n}-1) d_{m}(n) r_m(n+kL-L)}^2.
\end{equation*}

Note that
$\abs{ \breve{f}(t) -f(t)}^2=\abs{ u_1(t) +u_2(t)}^2$ and by the independence  of noise,
we can  use the integrations of $u_1$ and $u_2$ to express the    expectation of  MSE; that is,
\begin{align}
   &\mathbb{E} \left(  \frac{1}{2\pi} \int_{0}^{2\pi}  \abs{ \breve{f}(t) -f(t)}^2 dt\right)\nonumber\\
  = &\mathbb{E} \left(  \frac{1}{2\pi} \int_{0}^{2\pi}  \abs{ u_1(t) + u_2(t)}^2 dt\right)\nonumber   \\
  = &\frac{1}{2\pi} \int_{0}^{2\pi}  \abs{u_1(t)}^2  dt+ \mathbb{E}\left(\frac{1}{2\pi} \int_{0}^{2\pi}\abs{u_2(t)}^2 dt\right).\nonumber
\end{align}
It follows that
\begin{align}
   &\mathbb{E} \left(  \frac{1}{2\pi} \int_{0}^{2\pi}  \abs{ \breve{f}(t) -f(t)}^2 dt\right)\nonumber
 \\
 = & \sum_{k=1}^{M} \norm{\sum_{m=1}^{M} \widetilde{\mathbf{R}}_{m,k}\mathbf{D}_m{\bm\lambda}_m  - \widetilde{\mathbf{R}}_{m,k}\mathbf{d}_m}_2^2 +\frac{\sigma_{\epsilon}^2}{L }\sum_{m=1}^{M} \sum_{k=1}^{M} \norm{\widetilde{\mathbf{R}}_{m,k}{\bm\lambda}_m}_2^2,  \label{MSE2}
\end{align}
where
\begin{align*}
  \widetilde{\mathbf{R}}_{m,k} &  = \operatorname{diag}(r_m(N_1+kL-L),r_m(N_1+1+kL-L),\cdots,r_m(N_1-1+kL)), \\
   \mathbf{d}_m &  = (d_m(N_1),d_m(N_1+1),\cdots,d_m(N_1+L-1))^{\mathrm{T}}, \\
   \mathbf{D}_m&= \operatorname{diag} (\mathbf{d}_m),  \\
   {\bm\lambda}_m& = (\lambda_{m, N_1},\lambda_{m, N_1+1},\cdots,\lambda_{m,L+N_1-1})^{\mathrm{T}}.
\end{align*}
Suppose that the matrices and vectors involved in (\ref{MSE2}) are real-valued. To minimize the above expectation of  MSE, we denote (\ref{MSE2}) by $\Phi_2({\bm\lambda}_1,{\bm\lambda}_2,\cdots,{\bm\lambda}_M)$. Differentiating $\Phi_2$ with respect to ${\bm\lambda}_1,{\bm\lambda}_2,\cdots,{\bm\lambda}_M$ and solving
\begin{equation}\label{critical_point}
  \nabla_{{\bm\lambda}_1} \Phi_2  =0, \  \nabla_{{\bm\lambda}_3} \Phi_2  =0,\ \cdots, \ \nabla_{{\bm\lambda}_M} \Phi_2  =0,
\end{equation}
we can obtain the critical point. Since  $\Phi_2$ is a quadratic function, the critical point gives the unique solution   for the  optimization problem
\begin{equation}\label{opti_problem2}
   \operatorname{min} \Phi_2({\bm\lambda}_1,{\bm\lambda}_2,\cdots,{\bm\lambda}_M).
\end{equation}

Observe that the equation (\ref{critical_point}) can be expressed by the following system of linear  equations:
\begin{equation*}
  \bm\Psi  \bm\lambda = \bm \zeta,
\end{equation*}
where $\bm\Psi$ is a partitioned matrix with $M\times M$  blocks, $\bm\lambda $ and $\bm \zeta$  are partitioned  column vectors with $M$ blocks, and
 \begin{align*}
  {\bm\Psi}_{(m,n)}  & = \sum_{k=1}^{M} \mathbf{D}_m \widetilde{\mathbf{R}}_{m,k}\widetilde{\mathbf{R}}_{n,k}\mathbf{D}_n + \delta(m-n) \frac{\sigma_{\epsilon}^2}{L } \widetilde{\mathbf{R}}_{m,k}\widetilde{\mathbf{R}}_{n,k},\\
   {\bm\lambda}_{(m)} & ={\bm\lambda}_{m}, \\
  {\bm \zeta}_{(m)} & = \sum_{k=1}^{M} \sum_{n=1}^{M} \mathbf{D}_m\widetilde{\mathbf{R}}_{m,k}\widetilde{\mathbf{R}}_{n,k}\mathbf{d}_m,
 \end{align*}
 then the solution of (\ref{opti_problem2}) is given by
 \begin{equation}\label{opti_sol_2}
    \bm\lambda^* =  \bm\Psi^{-1}  \bm \zeta.
 \end{equation}

If (\ref{MSE2}) contains  complex-valued matrices or vectors, then (\ref{opti_problem2}) becomes an optimization problem with complex variables.  In this case,   the objective function can be rewritten as a function of the real and imaginary parts of its complex argument. For any complex matrix $\mathbf{A} = \mathbf{A}_r + \mathbf{A}_i \qi\in \mathbb{C}^{p\times q}$,  $\mathbf{A}_r,\mathbf{A}_i\in \mathbb{R}^{p\times q}$, we define an  operator $\Gamma: \mathbb{C}^{p\times q}\to \mathbb{R}^{2p\times 2q}$ such that
\begin{equation*}
  \Gamma(\mathbf{A}) = \begin{bmatrix}
 \mathbf{A}_r & - \mathbf{A}_i  \\
 \mathbf{A}_i &\mathbf{A}_r
	\end{bmatrix}\in \mathbb{R}^{2p\times 2q}.
\end{equation*}
 Similarly, for  any complex vector $\mathbf{b} = \mathbf{b}_r + \mathbf{b}_i \qi\in \mathbb{C}^{p}$,  $\mathbf{b}_r,\mathbf{b}_i\in \mathbb{R}^{p}$, we define an  operator $\gamma: \mathbb{C}^{p}\to \mathbb{R}^{2p}$ such that
\begin{equation*}
  \gamma(\mathbf{b}) = \begin{bmatrix}
 \mathbf{b}_r    \\
 \mathbf{b}_i
	\end{bmatrix}\in \mathbb{R}^{2p}.
\end{equation*}
It can be verified that $\Gamma$ and $\gamma$ have the following properties.
\begin{enumerate}
  \item [(1)]  $\Gamma$ and $\gamma$ are   invertible.
  \item  [(2)]   $\Gamma(c_1\mathbf{A}+c_2\mathbf{B}) = c_1 \Gamma( \mathbf{A})+ c_2\Gamma( \mathbf{B})$,   $\forall c_1,c_2\in\mathbb{R},\forall\mathbf{A},\mathbf{B}\in  \mathbb{C}^{p\times q}$.
  \item [(3)]  $\|\gamma(\mathbf{b})\|_2  =\| \mathbf{b}\|_2 $, $\forall \mathbf{b}\in\mathbb{C}^{p}$.
  \item [(4)] $\gamma(\mathbf{A}\mathbf{b}) = \Gamma(\mathbf{A})\gamma(\mathbf{b})$, $\forall \mathbf{A}\in\mathbb{C}^{p\times q}, \forall \mathbf{b}\in \mathbb{C}^{p}$.
\end{enumerate}
 By using the defined operators $\Gamma$ and $\gamma$, we can reformulate the objective function as
 \begin{equation*}
    \sum_{k=1}^{M} \norm{\sum_{m=1}^{M} \Gamma\left(\widetilde{\mathbf{R}}_{m,k}\mathbf{D}_m\right)\gamma\left({\bm\lambda}_m\right)  - \gamma\left( \widetilde{\mathbf{R}}_{m,k}\mathbf{d}_m\right)}_2^2 +\frac{\sigma_{\epsilon}^2}{L }\sum_{m=1}^{M} \sum_{k=1}^{M} \norm{\Gamma\left(\widetilde{\mathbf{R}}_{m,k}\right)\gamma\left({\bm\lambda}_m\right)}_2^2.
 \end{equation*}
 Thus the optimization problem  can still be solved by  the standard method.
\begin{remark}
   We have provided the closed-form solution for the optimization problem (\ref{opti_problem2}). If the size of the data   is very large,  some existing  software packages,   for example,  CVX \cite{cvx}, can  also be used to solve the optimization problem.
\end{remark}

The formula (\ref{opti_sol_2}) gives the optimal values for the parameters of pre-filtering. To compute $\bm\lambda^*$,  we have to estimate the value of $\mathbf{d}_m$ ($1\leq m\leq M$) from the noisy multichannel samples. In fact, if $f(t)$ is bandlimited with bandwidth $\leq N_s$, then $\mathbf{d}_\epsilon$ is an unbiased estimation for $(\mathbf{d}_1;\mathbf{d}_2;\cdots;\mathbf{d}_M )$ (see more details in Section 3 of \cite{cheng2019FFTMCI}). It is obvious that  the spectral density  of a function wouldn't change with the sampling schemes in the multichannel reconstruction. For $\mathbf{d}_m$, however, it changes with the types and amounts of samples even for a fixed $f(t)$, as shown in Table \ref{table2}. Different from post-filtering, the bandwidth of the reconstructed signal $\breve{f}$ based on pre-filtering equals the number of samples,  it can't be   adjusted to be equal to the bandwidth of $f$ by  taking specific parameters.

\begin{table}[ht]
\centering
\caption{The values of $\mathbf{d}_m$  ($1\leq m\leq M$) for the test function defined by (\ref{testfunc1}) under different sampling schemes.} \vspace{0.2cm}
\begin{tabular}{|l|l|}
	\hline
  Sampling schemes  & The true values of $\mathbf{d}_m$    \\
	\hline
  $L=6$, $M=1$, $f$	& $\mathbf{d}_1=(1+\qi,2-\qi,1,2+\qi,1-\qi,0)$    \\
	\hline
 $L=3$, $M=2$, $f$, $f'$		& $\mathbf{d}_1=(3+2\qi,3-2\qi,1)$,  $\mathbf{d}_2=(1,1,0)$  \\
	\hline
$L=3$, $M=2$, $f$, $\mathcal{H}f$		& $\mathbf{d}_1=(3+2\qi,3-2\qi,1)$,  $\mathbf{d}_2=(-\qi,\qi,0)$  \\
	\hline
$L=8$, $M=1$, $f$ & $\mathbf{d}_1=(0,1+\qi,2-\qi,1,2+\qi,1-\qi,0,0)$  \\
	\hline
 $L=4$, $M=2$, $f$, $f'$		& $\mathbf{d}_1=(2+\qi,2,2-\qi,1)$,  $\mathbf{d}_2=(2\qi-1,4,-1-2\qi,0)$  \\
	\hline
 $L=4$, $M=2$, $f$, $\mathcal{H}f$		& $\mathbf{d}_1=(2+\qi,2,2-\qi,1)$,  $\mathbf{d}_2=(1-2\qi,-2, 1+2\qi,0)$  \\
	\hline
\end{tabular}\label{table2}
\end{table}

\subsection{A simple comparison between pre-filtering and post-filtering}

It can be seen that   the model of pre-filtering  is more intricate than that of post-filtering. A direct strategy to simplify pre-filtering is     enforcing
\begin{equation}\label{specical_case}
{\bm\lambda}_1={\bm\lambda}_2=\cdots={\bm\lambda}_M.
\end{equation}
 In other words, pre-filtering every channel of samples in the same way. Under this restriction, the objective function becomes
\begin{equation*}
  \Phi_2(\bm\lambda_0)=\sum_{k=1}^{M} \norm{\sum_{m=1}^{M} \widetilde{\mathbf{R}}_{m,k}\mathbf{D}_m{\bm\lambda}_0   - \widetilde{\mathbf{R}}_{m,k}\mathbf{d}_m}_2^2 +\frac{\sigma_{\epsilon}^2}{L }\sum_{m=1}^{M} \sum_{k=1}^{M} \norm{\widetilde{\mathbf{R}}_{m,k}{\bm\lambda}_0 }_2^2,
\end{equation*}
where ${\bm\lambda}_0 =(\lambda_{N_1},\lambda_{N_2},\cdots,\lambda_{L+N_1-1}) $.
By the definition of $r_m$ , we have
\begin{equation*}
\sum_{m=1}^{M} b_m(n+jL-L)r_m(n+kL-L) =\delta(j-k),\quad \forall n\in I_1.
\end{equation*}
It follows that
\begin{align*}
    & \sum_{m=1}^{M}d_m(n)r_m(n+kL-L) \\
  = &  \sum_{m=1}^{M} \left( \sum_{j=1}^{M} a(n+jL-L)b_m(n+jL-L)\right)r_m(n+kL-L)\\
  =&\sum_{j=1}^{M} a(n+jL-L)\sum_{m=1}^{M} b_m(n+jL-L)r_m(n+kL-L) \\
  =&\sum_{j=1}^{M} a(n+jL-L)\delta(j-k)\\
  =& a(n+kL-L),\quad \forall n\in I_1.
\end{align*}
Therefore
\begin{equation*}
  \sum_{m=1}^{M} \widetilde{\mathbf{R}}_{m,k}\mathbf{d}_m =\left(  a(N_1+kL-L), a(N_1+1+kL-L),\cdots, a(N_1+kL-1) \right)^{\text{T}},
\end{equation*}
 \begin{equation*}
   \widetilde{\mathbf{R}}_{m,k}\mathbf{D}_m =  \operatorname{diag} \left(\sum_{m=1}^{M} \widetilde{\mathbf{R}}_{m,k}\mathbf{d}_m \right).
 \end{equation*}
Thus the objective function can be simplified as
\begin{equation*}
   \Phi_2(\bm\lambda)=\sum_{k=1}^{M}\sum_{n\in I_1} \abs{a(n+kL-L)(\lambda_n-1)}^2 + \frac{\sigma_\epsilon^2}{L}\sum_{m=1}^{M}\sum_{k=1}^{M}\sum_{n\in I_1}\abs{r_m(n+kL-L)\lambda_n}^2.
\end{equation*}
Recalling the objective function of post-filtering  defined by (\ref{opti_pro_post}), we see that the pre-filtering under the condition that ${\bm\lambda}_1={\bm\lambda}_2=\cdots={\bm\lambda}_M$ is a special case of  post-filtering.  This means that pre-filtering and post-filtering have a nontrivial intersection.

 It is noted that if $M=1$, the condition defined by (\ref{specical_case}) is satisfied naturally.
When $M=1$, the objective function of post-filtering  becomes
\begin{equation*}
   \Phi_1(\bm\beta) =  \sum_{k=K_1}^{K_2} \abs{a(k) (\beta_k-1)}^2 + \frac{\sigma_\epsilon^2 }{L } \sum_{k=K_1}^{K_2} \abs{r_1(k,\mathrm{Type},N_s)\beta_k}^2 .
\end{equation*}
For the sampling scheme F1, we have $r_1(k,\mathrm{Type},N_s)=1$, then
\begin{equation*}
   \Phi_1(\bm\beta,\mathrm{F1}) =  \sum_{k=K_1}^{K_2} \abs{a(k) (\beta_k-1)}^2 + \frac{\sigma_\epsilon^2 }{L } \sum_{k=K_1}^{K_2} \abs{ \beta_k}^2 .
\end{equation*}
It follows that
\begin{equation*}
 {\arg\min} ~\Phi_1(\bm\beta,\mathrm{F1}) = \left(  \abs{a( k )}^2\left(\abs{a(k)}^2+\frac{\sigma_\epsilon^2}{L} \right)^{-1}:~~ K_1 \leq k \leq K_2\right).
\end{equation*}
It   is actually the Wiener filter \cite{chen2006new}.
By the analysis of Section \ref{S3}, we see that if $M=1$, the post-filtering is equivalent to the pre-filtering followed by a ideal low-pass filtering. Therefore
\begin{equation*}
 {\arg\min} ~\Phi_2(\bm\lambda,\mathrm{F1}) = \left(  \abs{a( k )}^2\left(\abs{a(k)}^2+\frac{\sigma_\epsilon^2}{L} \right)^{-1}:~~ N_1 \leq k \leq N_2\right).
\end{equation*}
From the above discussion, we conclude that the   pre-filtering  and the  post-filtering are two extensions of the Wiener filter in the multichannel reconstruction setting.

\begin{figure}
  \centering
  \includegraphics[width=8cm]{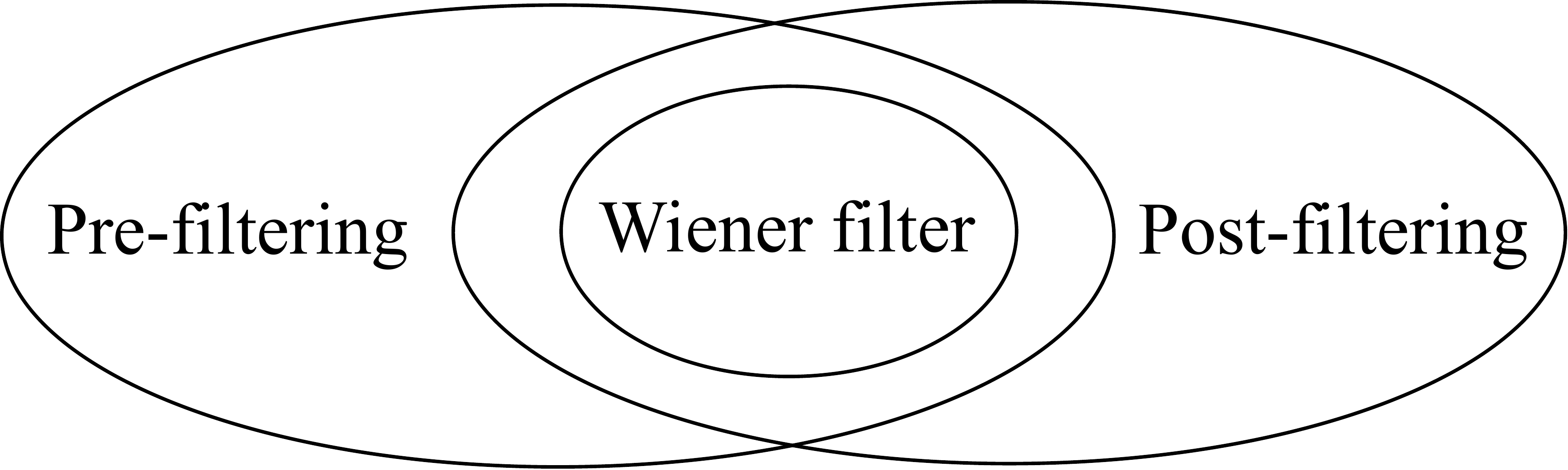}
  \caption{Inclusion relationship of filtering processes}\label{inclusion}
\end{figure}

\subsection{Regularized approximation}

If the multichannel samples $\mathbf{s}_1,\mathbf{s}_2,\cdots,\mathbf{s}_M$ are noiseless.  The goal of reconstruction is to find a function $  \tilde{f}(t)$ such that
 the residual
 $ \sum_{m=1}^{M} \norm{\mathbf{\tilde{g}}_m -\mathbf{s}_m}$
 is small, where
\begin{equation*}
\mathbf{\tilde{g}}_m = \left(\tilde{f}*h_m(t_0),\tilde{f}*h_m(t_1),\cdots, \tilde{f}*h_m(t_{L-1})\right)^{\text{T}},\quad t_p=\tfrac{2\pi p}{L},1\leq m \leq M.
\end{equation*}
That is, to solve the  optimization problem:
\begin{equation*}
   \mathrm{minimize}   \  \  \sum_{m=1}^{M} \norm{\mathbf{\tilde{g}}_m -\mathbf{s}_m}
\end{equation*}
with respect to $\tilde{f}$.
If the multichannel samples are noisy, it is desirable that a small variation in $\mathbf{s}_m$ ($1\leq m \leq M$) would not cause a large variation in $\tilde{f}$. Generally, however, the optimal value  of the above optimization problem is sensitive to the data $\mathbf{s}_m$ ($1\leq m \leq M$). A common used approach to deal with the sensitivity is to  add a penalty term $ \varphi(\tilde{f})$ in  optimization \cite{Ramani2008nonideal}.
Consequently, the reconstruction problem can be described as a optimization problem of minimizing a weighted sum of   two objectives:
\begin{equation*}
  \mathrm{minimize} \ \     \sum_{m=1}^{M} \norm{\mathbf{\tilde{g}}_m -\mathbf{s}_m} + \alpha \varphi(\tilde{f}).
\end{equation*}

Suppose that
\begin{equation*}
\tilde{f}(t) = \sum_{n=N_1}^{N_2}x(n) e^{\qi nt},
\end{equation*}
we rewrite the above optimization problem as follows:
\begin{equation*}
   \mathrm{minimize} \ \     \sum_{m=1}^{M} \norm{\mathbf{C}_m \mathbf{x} -\mathbf{s}_m}_2^2 + \alpha \varphi(\mathbf{x} ),
\end{equation*}
where $\mathbf{x}=\left(x(N_1),x(N_1+1),\cdots,x(N_2)\right)^{\text{T}}$ and
\begin{equation*}
 \mathbf{C}_m =   \begin{bmatrix}
b_m(N_1)e^{\qi N_1t_0} & b_m(N_1+1)e^{\qi (N_1+1)t_0}  &\cdots &b_m(N_2)e^{\qi  N_2 t_0} \\
b_m(N_1)e^{\qi N_1t_1} &  b_m(N_1+1)e^{\qi (N_1+1)t_1}  &\cdots &b_m(N_2)e^{\qi  N_2 t_1} \\
\vdots&\vdots& ~ & \vdots \\
b_m(N_1)e^{\qi N_1t_{L-1}} &  b_m(N_1+1)e^{\qi (N_1+1)t_{L-1}}  &\cdots &b_m(N_2)e^{\qi  N_2 t_{L-1}}
\end{bmatrix}.
\end{equation*}
The key problem is how to design the penalty term $\varphi$ that brings robustness to the reconstruction.

 Let
\begin{equation*}
  \mathbf{W}_\eta = \operatorname{diag}\left(1+\abs{N_1}^\eta,1+\abs{N_1+1}^\eta,\cdots,1+\abs{N_2}^\eta \right)
\end{equation*}
be a weighting matrix, we consider the following penalties:
\begin{equation*}
  \varphi_1(\mathbf{x}) =\sigma_{\epsilon}^2\norm{\mathbf{W}_{\eta}\mathbf{x}}_1,\quad  \varphi_2(\mathbf{x}) =\sigma_{\epsilon}^2\norm{\mathbf{W}_{\eta}\mathbf{x}}_2^2.
\end{equation*}
The role of $\mathbf{W}_\eta$ is to  suppress the high frequency to achieve smoothing, and the parameter  $\eta$ is used to control the amount of smoothness.
The $l_1$-norm appeared in $\varphi_1(\mathbf{x}) $ induces  sparsity \cite{donoho2006stable} for the spectral of $\tilde{f}$. The hypothesis of sparsity is based on the priori knowledge that the actual effective information of a signal $f$ is concentrated in a low-dimensional space even if    it is high dimensional.
The other penalty term $\varphi_2(\mathbf{x})$ is a common form  of regularization based on the $l_2$-norm, say Tikhonov regularization. This regularization results in a  quadratic optimization problem, thus  enabling an explicit solution.

By using the operators $\Gamma$ and $\gamma$, the objective function of the $l_2$ regularization  can be rewritten as
\begin{equation*}
      \sum_{m=1}^{M} \norm{\Gamma\left(\mathbf{C}_m\right) \gamma\left(\mathbf{x}\right) -\gamma\left(\mathbf{s}_m\right) }_2^2 + \alpha\sigma_{\epsilon}^2\norm{\Gamma\left(\mathbf{W}_{\eta}\right)\gamma\left(\mathbf{x}\right)}_2^2.
\end{equation*}
To solve  the optimization problem of the $l_1$ regularization with complex variables, we need to use the permutation matrix $\mathbf{P}_K$  corresponding to the permutation
\begin{equation*}
  (1,2,\cdots,K)\mapsto (1,\frac{K}{2}+1,2,\frac{K}{2}+2,\cdots,\frac{K}{2},K),
\end{equation*}
where $K$ is an even number. Let $\mathbf{y}=\mathbf{W}_{\eta} \mathbf{x}$, then the objective function
\begin{equation*}
    \sum_{m=1}^{M} \norm{\mathbf{C}_m \mathbf{x} -\mathbf{s}_m}_2^2 + \alpha\sigma_{\epsilon}^2 \norm{\mathbf{W}_{\eta}\mathbf{x}}_1
\end{equation*}
 becomes
 \begin{align*}
     &  \sum_{m=1}^{M} \norm{\mathbf{C}_m\mathbf{W}_{\eta}^{-1} \mathbf{y} -\mathbf{s}_m}_2^2 + \alpha\sigma_{\epsilon}^2 \norm{ \mathbf{y}}_1\\
     =&    \sum_{m=1}^{M} \norm{\Gamma\left(\mathbf{C}_m\mathbf{W}_{\eta}^{-1}\right) \gamma\left(\mathbf{y}\right) -\gamma(\mathbf{s}_m)}_2^2 + \alpha\sigma_{\epsilon}^2 \norm{ \mathbf{y}}_1
 \end{align*}
Let $\mathbf{y}^{(p)} =\mathbf{P}_K \gamma\left(\mathbf{y}\right) $ and we partition  $\mathbf{y}^{(p)} $ to be $(\mathbf{y}^{(p)}_1;\mathbf{y}^{(p)}_2;\cdots;\mathbf{y}^{(p)}_{\frac{K}{2}})$, where $\mathbf{y}^{(p)}_k \in \mathbb{R}^2$, $1\leq k\leq{\frac{K}{2}}$. Namely, we view $\mathbf{y}^{(p)}$ as a partitioned column vector and its $k$-th block is a two dimensional vector $\mathbf{y}^{(p)}_k $.  By the definition of $\mathbf{y}^{(p)}$  and using the operator $\gamma$, we have that
\begin{equation*}
  \norm{ \mathbf{y}}_1 = \sum_{k=1}^{K/2}\norm{\mathbf{y}^{(p)}_k}_2.
\end{equation*}
Thus, the objective function of the $l_1$ regularization can be rewritten
as
\begin{equation*}
  \sum_{m=1}^{M} \norm{\Gamma\left(\mathbf{C}_m\mathbf{W}_{\eta}^{-1}\right) \mathbf{P}_K^{\text{T}} \mathbf{y}^{(p)} -\gamma(\mathbf{s}_m)}_2^2 + \alpha\sigma_{\epsilon}^2 \sum_{k=1}^{K/2}\norm{\mathbf{y}^{(p)}_k}_2.
\end{equation*}
It is actually a sum-of-norms regularization problem \cite{Ohlsson2010} and it can be solved by the Alternating Direction Method of Multipliers (ADMM) framework  \cite{Boyd2011}.

\section{Theoretical convergence analysis of post-filtering}\label{S4}

First, we assume that  $f\in B_\mathbf{K}, I^{\mathbf{K}}=\{k,K_1\leq k\leq K_2\}$ and $\mu(I^{\mathbf{K}})\leq \mu(I^{\mathbf{N}})=N_s$, where $N_s$ is the total number of samples. Let $w(t)   \in B_\mathbf{K}$.
The minimum value of  $\Phi_1(\bm\beta,\mathrm{F1})$  is equal to
\begin{equation*}
  \Phi_1(\bm\beta^*,\mathrm{F1})  = \sum_{k=K_1}^{K_2} \frac{\abs{a(k)}^2\sigma_\epsilon^2}{\abs{a(k)}^2\cdot L+ \sigma_\epsilon^2}.
\end{equation*}
It follows that
\begin{equation*}
  \lim\limits_{L\to \infty} \Phi_1(\bm\beta^*,\mathrm{F1}) = 0.
\end{equation*}
Similarly,
\begin{equation*}
  \lim\limits_{L\to \infty} \Phi_2(\bm\lambda_0^*,\mathrm{F1}) = 0.
\end{equation*}
Using the previously computed $r_m$ for FH2 and FD2, we have that
\begin{equation*}
\Phi_1(\bm\beta^*,\mathrm{FH2}) = \frac{2\abs{a(0)}^2\sigma_{\epsilon}^2}{\abs{a(0)}^2 L+2\sigma_{\epsilon}^2} + \sum_{K_1\leq k\leq K_2,k\neq 0} \frac{\abs{a(k)}^2\sigma_{\epsilon}^2}{2\abs{a(k)}^2L+\sigma_{\epsilon}^2},
\end{equation*}
and
\begin{equation*}
\Phi_1(\bm\beta^*,\mathrm{FD2}) = \sum_{k=K_1}^{K_2}\frac{\abs{a(k)}^2(1+(L-\abs{k})^2)\sigma_{\epsilon}^2}{\abs{a(k)}^2L^3+(1+(L-\abs{k})^2)\sigma_{\epsilon}^2}.
\end{equation*}
We see that both $\Phi_1(\bm\beta^*,\mathrm{FH2})$ and $\Phi_1(\bm\beta^*,\mathrm{FD2}) $ are convergent to $0$ as $L\to \infty$.

An important choice of $w(t)$ for post-filtering is $D(t,K_1,K_2)$ defined by (\ref{Dkernel}), an analogue of the ideal low-pass filter. Note that
\begin{equation*}
f*D(\cdot,K_1,K_2)(t) =f(t)
\end{equation*}
provided that $f\in B_\mathbf{K}$. It follows that
\begin{align*}
\widetilde{f}(t,N_s,\mathbf{K})-f(t)  & =  f_{\mathbf{N},\epsilon}*D(\cdot,K_1,K_2) (t)- f*D(\cdot,K_1,K_2)(t)\\
& = \left( (f_{\mathbf{N},\epsilon}-f)*D(\cdot,K_1,K_2)\right) (t)\\
& = \frac{1}{L}\sum_{m=1}^{M} \sum_{p=0}^{L-1} \epsilon_{m,p}\left( y_m *D(\cdot,K_1,K_2) \right)(t-\tfrac{2\pi p}{L}).
\end{align*}
Since $\{\epsilon_{m,p}\}$ is an i.i.d. noise process, we have that
\begin{align*}
& \mathbb{E}\abs{\widetilde{f}(t,N_s,\mathbf{K})-f(t)}^2 \\
= &\frac{1}{L^2}\sum_{m=1}^{M} \sum_{p=0}^{L-1}\abs{\left( y_m *D(\cdot,K_1,K_2) \right)(t-\tfrac{2\pi p}{L})}^2  \mathbb{E} [\epsilon_{m,p}^2]\\
=& \frac{1}{L^2}\sum_{m=1}^{M} \sum_{p=0}^{L-1}\abs{\left( y_m *D(\cdot,K_1,K_2) \right)(t-\tfrac{2\pi p}{L})}^2   \sigma_{\epsilon}^2.
\end{align*}
Therefore
\begin{align*}
&\mathbb{E} \left(\frac{1}{2\pi} \int_{0}^{2\pi} \abs{ \widetilde{f}(t,N_s,\mathbf{K})-f(t)}^2 dt\right) \\
= &  \frac{1}{2\pi}\int_{0}^{2\pi}  \mathbb{E}\abs{\widetilde{f}(t,N_s,\mathbf{K})-f(t)}^2 dt\\
=& \frac{\sigma_{\epsilon}^2}{L^2}\sum_{m=1}^{M} \sum_{p=0}^{L-1} \frac{1}{2\pi}\int_{0}^{2\pi} \abs{\left( y_m *D(\cdot,K_1,K_2) \right)(t-\tfrac{2\pi p}{L})}^2 dt\\
=& \frac{\sigma_\epsilon^2 }{L } \sum_{m=1}^{M}   \norm{y_m *D(\cdot,K_1,K_2)}_2^2= \frac{\sigma_\epsilon^2 }{L} \sum_{m=1}^{M} \sum_{n\in I^{\mathbf{K}}}\abs{r_m(n,\mathrm{Type},N_s)}^2 .
\end{align*}

We examine the accuracy of the reconstructed signal $\widetilde{f}(t,N_s,\mathbf{K})$
for   sampling schemes F1, FH2 and FD2. Without loss of generality, let $\mu (I^{\mathbf{K}})=2K_2$,  $K_1=1-K_2$, we have that
\begin{equation*}
\frac{\sigma_\epsilon^2 }{L} \sum_{m=1}^{M} \sum_{n\in I^{\mathbf{K}}}\abs{r_m(n,\mathrm{F1},N_s)}^2   = \frac{2K_2 \sigma_\epsilon^2}{N_s},
\end{equation*}
\begin{equation*}
\frac{\sigma_\epsilon^2 }{L} \sum_{m=1}^{M} \sum_{n\in I^{\mathbf{K}}}\abs{r_m(n,\mathrm{FH2},N_s)}^2   = \frac{(2K_2+3) \sigma_\epsilon^2}{N_s},
\end{equation*}
\begin{equation*}
\frac{\sigma_\epsilon^2 }{L} \sum_{m=1}^{M} \sum_{n\in I^{\mathbf{K}}}\abs{r_m(n,\mathrm{FD2},N_s)}^2   = \frac{4\sigma_\epsilon^2}{N_s} + \frac{56K_2\sigma_\epsilon^2}{3N_s^3} +  \frac{16K_2^3\sigma_\epsilon^2}{3N_s^3} - \frac{8K_2\sigma_\epsilon^2}{N_s^2}.
\end{equation*}
It is easy to see that $\frac{\sigma_\epsilon^2 }{L} \sum_{m=1}^{M} \sum_{n\in I^{\mathbf{K}}}\abs{r_m(n,\mathrm{Type},N_s)}^2 \to 0$ as $N_s\to \infty$. As for different sampling schemes,    FD2 performs better than F1   and FH2.

\begin{remark}
To  examine the convergence property of the optimal post-filtering for the sampling schemes in addition to F1, FH2 and FD2, one can consider firstly the post-filtering by the Dirichlet kernel for its simplicity, since the expectation of MSE for the optimal post-filtering is always no larger than that of   post-filtering by the Dirichlet kernel.
\end{remark}

\begin{figure}[t]
	\centering
	\includegraphics[width=12.0cm]{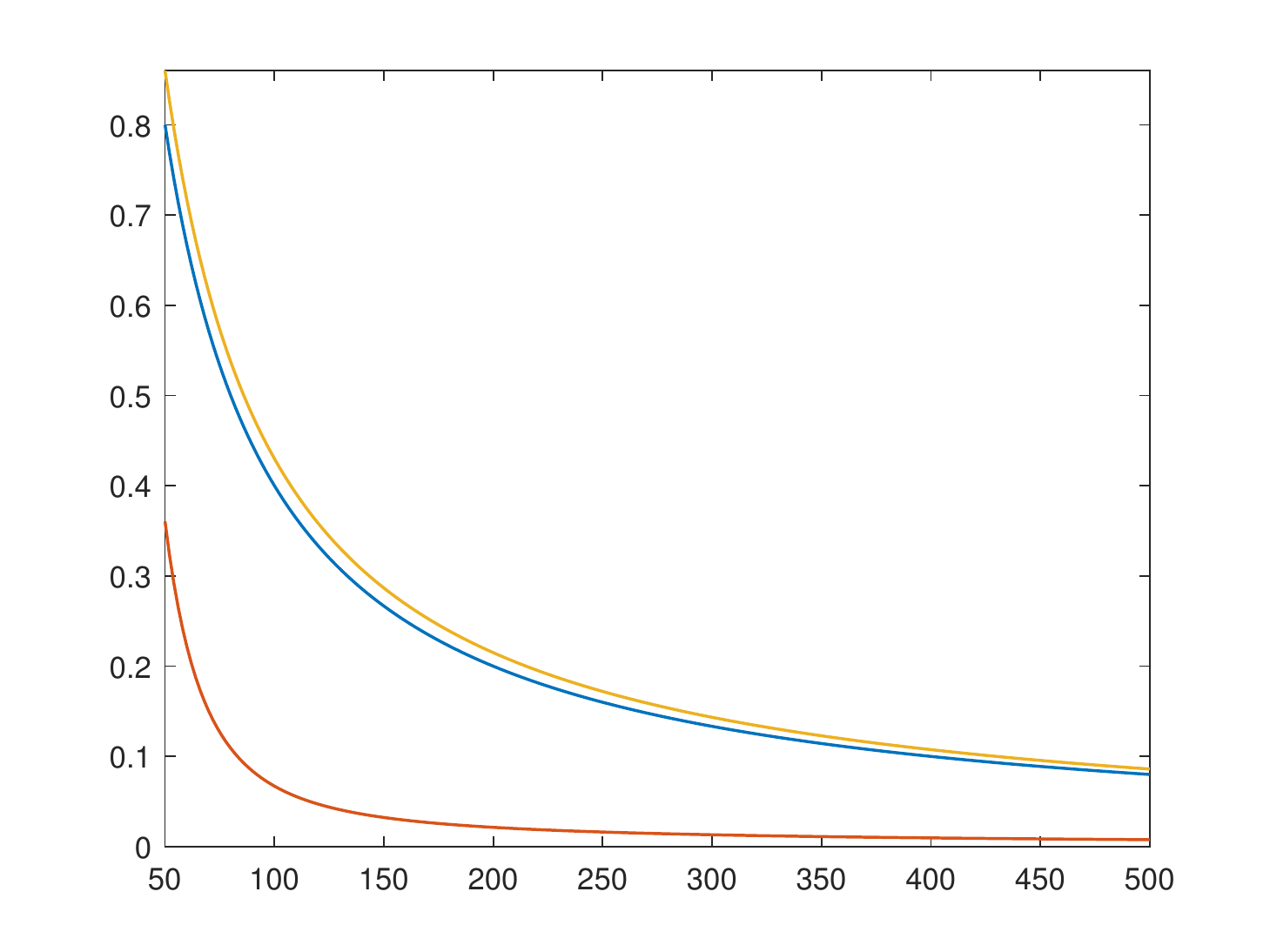}
	\caption{Plot of $\frac{\sigma_\epsilon^2 }{L} \sum_{m=1}^{M} \sum_{n\in I^{\mathbf{K}}}\abs{r_m(n,\mathrm{Type},N_s)}^2$  with $K_2=20$ for F1 (blue), FH2 (yellow) and FD2 (orange). The horizontal axis  represents  the total number of samples  $N_s$. }\label{lowpasserror}
\end{figure}

To analyze the convergence of post-filtering for the non-bandlimited signal, we need the following lemmas.
 \begin{lemma}\label{lem41}
   If $f\in C^j(\mathbb{T})$, then there exists a constant $\gamma$ such that $\abs{a(n)}\leq {\gamma }/{\abs{n}^j}$ for $n\neq 0$.
 \end{lemma}

 \begin{lemma}\label{lem42}
   Let $c_k=\sum_{n=k+1}^{\infty} {n^{-\alpha}}$ with $\alpha>1$, then $$(\alpha-1)^{-1}(k+1)^{1-\alpha}<c_k<(\alpha-1)^{-1} k^{1-\alpha}$$ for all $k\in\mathbb{Z}^{+}$.
 \end{lemma}
 \begin{proof}
 Since $\alpha>1$, then $ x^{-\alpha}$ is a monotonic decreasing function on $(0,+\infty)$. It follows that
    \begin{equation*}
      \int_{k+1}^{\infty} x^{-\alpha}dx <  \sum_{n=k+1}^{\infty}  \frac{1}{n^\alpha}<\int_{k}^{\infty} x^{-\alpha}dx,\quad \forall k\in\mathbb{Z}^{+}.
   \end{equation*}
Note that    $\int_{k+1}^{\infty} x^{-\alpha}dx =(\alpha-1)^{-1}(k+1)^{1-\alpha} $ and $\int_{k}^{\infty} x^{-\alpha}dx=(\alpha-1)^{-1} k^{1-\alpha}$, the proof is complete.
 \end{proof}

If $f$ is non-bandlimited, the error of  the reconstruction for $f$ by  post-filtering  comes from not only noise but also aliasing. Let $c_k =\sum_{\abs{n}\geq k+1}^{\infty} \abs{a(n)}^2$, where $a(n)$ is the Fourier coefficient  of $f$. Then $c_k$  tends to $0$ as $k\to \infty$ under the assumption that  $f\in L^2(\mathbb{T})$. By Lemmas \ref{lem41} and \ref{lem42}, we see that the convergence rate of  $c_k \to 0$ can be very fast if $f$ is smooth.  It is known that if  the number of    samples $N_s$ is  sufficient large such that $c_K$ ($K=(N_s-1)/2$) is    sufficient small, then the aliasing error is negligible. Therefore, if we take sufficient large $K_1,K_2$ in   post-filtering, the   aliasing error can be sufficient small.

For  pre-filtering and   regularized approximation, there is no guarantee that the expectation of MSE will tend to $0$ when the  number of samples goes to infinity. Nonetheless, it doesn't mean that
the performance of   post-filtering  is better than  pre-filtering and   regularized approximation in signal reconstruction when the   number of samples is finite.
It is important to clarify whether the reconstruction results using  smoothing corrections would converge to the original signal as the number of samples tends to infinity.
In practice, however, the total number of samples is finite. Therefore, we need to examine which of the aforementioned methods  can achieve preferable reconstruction results  under the same sample set of limited size. Besides, it is necessary to test  whether combining    pre-filtering and   post-filtering could give   some distinctive results. In the next section, the proposed smoothing and regularization strategies  are verified    comprehensively by  a number of numerical simulations.

\section{Numerical simulations}\label{S5}

In the last section we introduced three smoothing strategies to reduce  the effect of noise  in  multichannel reconstruction. To provide a more intuitive understanding of the above theoretical analysis, we give  a number of  examples to show the exact formulas of the proposed smoothing strategies in some concrete sampling schemes.

Recall $\phi(z)$ defined by (\ref{testfunc}), we use $f(t) = \phi(e^{\qi t})$ as the test function. Obviously, $f(t)$ is non-bandlimited. To use the proposed post-filtering, a suitable choice of $K_1$ and $K_2$ is needed. Given a set of noisy multichannel samples,   an estimate for $\abs{a(n)}^2$ is obtained  by applying the estimation method  for the  spectral density (see Section \ref{s312}).
Then we take  the smallest possible  $\widetilde{K}_1,\widetilde{K}_2$ (absolute value)  such that
\begin{equation}\label{K1K2}
 \frac{\sum_{n=\widetilde{K}_1}^{\widetilde{K}_2}  \tilde{A}(n) }{\sum_{n}  \tilde{A}(n) } \geq 0.9,
\end{equation}
where $\tilde{A}(n)$ is the estimate for  $\abs{a(n)}^2$. The selection of $\widetilde{\mathbf{K}}=\{n: \widetilde{K}_1\leq n \leq \widetilde{K}_2\}$ is  diagramed in Figure \ref{data_analysis}.
If the number of samples  is sufficiently large, the bandwidth of the post-filter can be   increased properly. Thus   the final selection of the bandwidth for post-filtering is
\begin{equation*}
  \mu(I^{\mathbf{K}}) = \max\{2\sqrt{N_s},\mu(I^{\widetilde{\mathbf{K}}})\}.
\end{equation*}

As for the parameters of $l_1$ and $l_2$ regularization, the experiential  values are $\eta =1.2$ and $\alpha=1$. It is hard to make out which set of parameters    $\eta,\alpha$ is the best due to the limited information of the signal acquired. But they can be tuned using   the estimated  spectral density. Roughly speaking, the faster decline of $\tilde{A}(n)$, the larger $\eta$; the  sparser of $\tilde{A}(n)$, the larger $\alpha$.

\begin{figure}
  \centering
  \includegraphics[width=10cm]{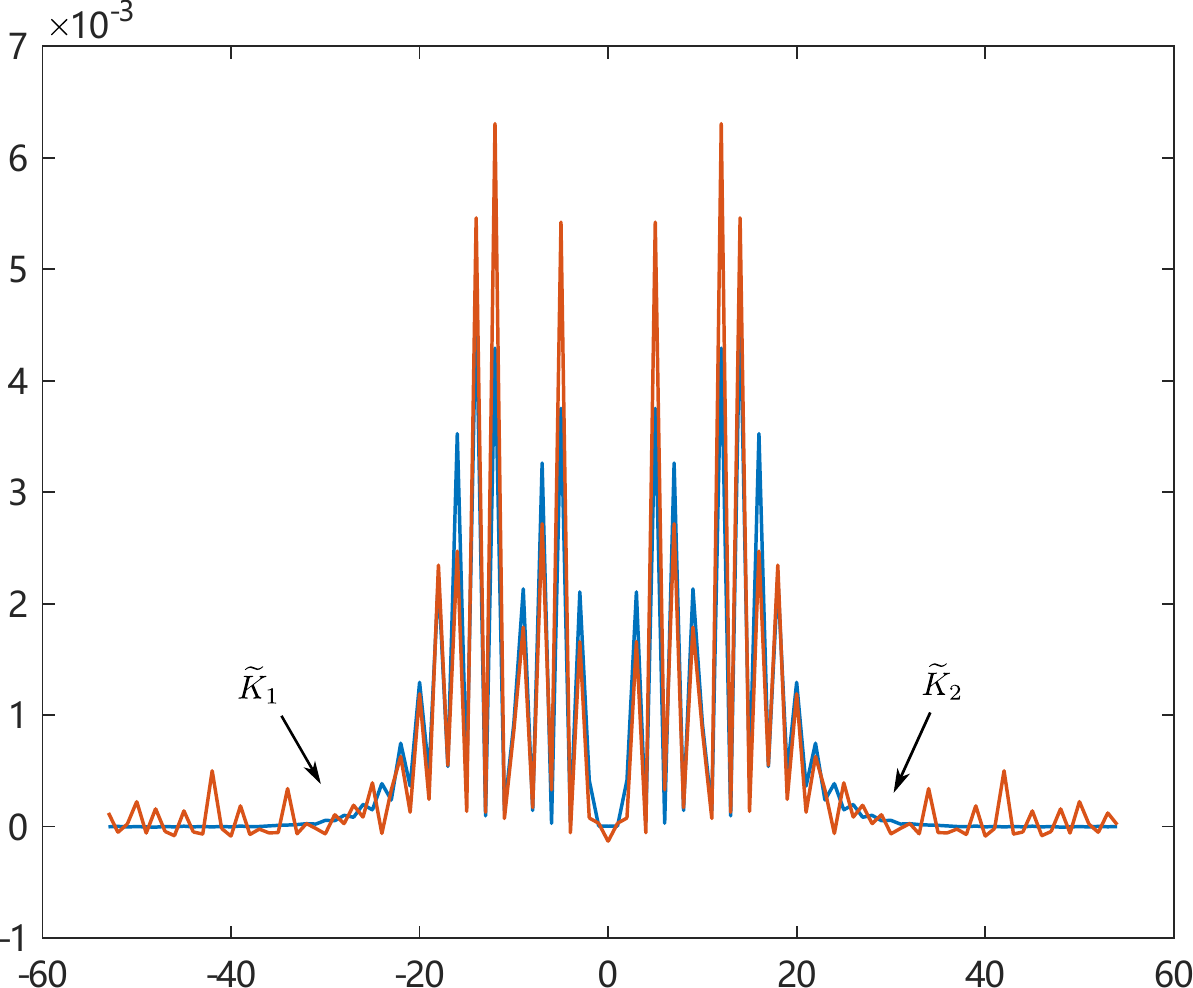}
  \caption{The blue line is the true value of $\abs{a(n)}^2$. The  orange line  is the estimated  spectral density $\tilde{A}(n)$   from $54$ noisy samples of $f$ and $54$ noisy samples of $\mathcal{H}f$.  The experiential values for $\widetilde{K}_1,\widetilde{K}_2$  can be set as the inflection points where $\tilde{A}(n)$  tapers off to zero.}\label{data_analysis}
\end{figure}

The experiments are conducted to compare the reconstructed results by the proposed methods under the sampling schemes of F1, FH2 and FD2.   The Gaussian noise with   standard deviation $0.05$ are superposed on the    samples of F1 and FH2. While the  samples of FD2 are   corrupted by the Gaussian noise with   standard deviation $0.1$. The expectation of mean square error (EMSE)  is approximated by  the average of mean square error over $10000$ times experiments. The  EMSEs for the  reconstructions in different situations  are listed in Table \ref{table3}, \ref{table4} and \ref{table5}. Besides,  Table \ref{table6} reports the running time for the reconstruction under the sampling scheme of FD2  in detail. Furthermore, the reconstructed results by the proposed methods under the sampling scheme  of   FD2 with $N_s = 1248$ is From the results we make the following conclusions.
\begin{enumerate}
  \item The experimental results prove that all the  reconstruction methods  based on smoothing or regularization  have improved in terms of accuracy    compared to the  MCI. Moreover, the errors  of reconstruction reduce  with the increase of the size of samples.
  \item If the sample size is small, the composite of pre-filtering and post-filtering  is effective. The double-stage  filtering can remove noise in the multichannel reconstruction  to a large extent.
  \item If the number of samples is relatively large,  the capability of  noise removal of pre-filtering is insufficient (see Figure \ref{Compare_results}). Whereas, using only the post-filtering can achieve a good performance in noise reduction.
 \item Although the errors of reconstruction by $l_1$ and $l_2$ regularization are relatively  small, the corresponding reconstructed signals look a bit rough  if no post-filtering is performed (see Figure \ref{Compare_results}). The $l_1$ or $l_2$ regularization  followed by post-filtering can achieve a slightly better performance than just using post-filtering. This improvement comes at the expense of high computational complexity. For example,  when $N_s= 1248$, $l_1$ regularization spend at least $1000$ times as much time as post-filtering on the reconstruction (see Table \ref{table6}).
\end{enumerate}

We have proved the convergence of post-filtering  theoretically. This implies that the reconstructed signal obtained by post-filtering will converge to  $f$ provided that  the true value of $\abs{a(n)}^2$ is available. In a  real application, however, we can only use the estimate of $\abs{a(n)}^2$, i.e., $\tilde{A}(n)$, to perform post-filtering. It has been shown that the more samples taken in reconstruction, the more accurate the spectral density estimate.   Thus it is not surprising that the error of reconstruction by post-filtering tends to zero as the number of samples tends to infinity in practice. Moreover, as can be seen from Figure \ref{Convergent_of_Post_filt} that the errors of reconstruction by pre-filtering, $l_1$ and $l_2$ regularization in combination with post-filtering are also convergent to $0$ as $N_s\to \infty$.

The above  analysis suggests that the composite of pre-filtering and post-filtering  is   preferable to the other methods when the sample size is small. Otherwise,   post-filtering a good choice because it could produce a fine result  in a relatively short period. If one does not give a high priority to the execution time, $l_1$  (or $l_2$ ) regularization  combined with post-filtering  has the capability to achieve  a remarkable  reconstruction result provided that suitable parameters are selected.


\section{Conclusion and discussion}\label{S6}

In this paper, we propose several smoothing and regularization based methods to reconstruct signal from its multichannel samples in the presence of noise. The  experiments show   that the proposed methods can be effectively used to  reduce the error caused by noise in multichannel reconstruction under limited size of samples. Moreover, we also develop the theoretical and experimental convergence analysis     to ensure the high precision in reconstruction   under large-scale samples.
 In addition,    some operators are devised  to  tackle the  optimization problems with complex variables such that  they can be solved by   the standard algorithms.

The proposed methods are established  under the MSE criterion  which is suitable for   the Gaussianity assumption of noise distribution.  Because maximum correntropy criterion  (MCC)  is independent of the noise distribution, it has been used to deal with the samples corrupted by non-Gaussian noise \cite{zou2018robust}.  To  cope with  various types of noises, therefore,
there is a great need to develop smoothing and regularization based
methods for multichannel reconstruction through minimizing correntropy induced metric (CIM). This
topic will be discussed in our further studies.

{ \begin{table}
\caption{The EMSEs of the reconstruction  from the noisy samples of $f$ and $\mathcal{H}f$ by different methods. The standard deviation of the noise is $\sigma_\epsilon = 0.05$. The parameters of $l_1$ regularization and $l_2$ regularization are $\eta= 1.2, \alpha = 1$.} \vspace{0.2cm}	\scriptsize \setlength\tabcolsep{3pt}
	\begin{tabular}{c|c c c c c c c c}
		\hline
		No.  &MCI  & MCI+ & Pre-filt & Pre-filt &$l_1$-Reg  &$l_1$-Reg+  &$l_2$-Reg&$l_2$-Reg+  \\
		of ${S}$	&  & Post-filt &+MCI  & +MCI+ &  & Post-filt &  & Post-filt  \\
		&  &  &  &  Post-filt &  &  & & \\
		\hline
$12$ & $1.9775\cdot 10^{-1}$  & $1.9552\cdot 10^{-1}$&   $1.9658\cdot 10^{-1}$  & $\mathbf{1.9511\cdot 10^{-1}}$ &  $1.9613\cdot 10^{-1}$  & $1.9613\cdot 10^{-1}$  & $1.9662\cdot 10^{-1}$ &  $1.9662\cdot 10^{-1}$\\
$24$ &$ 1.1846\cdot 10^{-1}$ &  $1.1686\cdot 10^{-1}$ &  $1.1647\cdot 10^{-1}$  & $\mathbf{1.1533\cdot 10^{-1}}$ &  $1.1770\cdot 10^{-1}$  & $1.1770\cdot 10^{-1}$ &  $1.1786\cdot 10^{-1}$ &  $1.1786\cdot 10^{-1}$\\
$36$  &$3.4946\cdot 10^{-2}$  & $\mathbf{3.3460\cdot 10^{-2}}$  &$ 3.4789\cdot 10^{-2}$  & $3.3491\cdot 10^{-2}$ &  $3.4673\cdot 10^{-2}$  & $3.4673\cdot 10^{-2}$ &  $3.4600\cdot 10^{-2}$ &  $3.4600\cdot 10^{-2}$\\
$48$  &$7.9627\cdot 10^{-3}$  & $7.3021\cdot 10^{-3}$ &  $7.7302\cdot 10^{-3} $ & $\mathbf{7.2938\cdot 10^{-3}}$&  $ 7.8393\cdot 10^{-3}$   &$7.8393\cdot 10^{-3}$&   $7.7912\cdot 10^{-3} $& $ 7.7912\cdot 10^{-3}$\\
$60 $ &$3.2912\cdot 10^{-3}$  & $3.0928\cdot 10^{-3} $& $\mathbf{2.9815\cdot 10^{-3}}$  &$ 3.1310\cdot 10^{-3}$&   $3.1743\cdot 10^{-3}$  & $3.1743\cdot 10^{-3} $ & $3.1252\cdot 10^{-3}$  & $3.1252\cdot 10^{-3}$\\
$72$   &$2.6287\cdot 10^{-3} $  &$2.4180\cdot 10^{-3}$ &  $\mathbf{2.2175\cdot 10^{-3}}$ & $ 2.4221\cdot 10^{-3} $ &  $2.4276\cdot 10^{-3}$ &  $2.4276\cdot 10^{-3}$ &$  2.4091\cdot 10^{-3}$  & $2.4091\cdot 10^{-3}$\\
$84  $ &$2.5423\cdot 10^{-3}$  & $2.0685\cdot 10^{-3}$ &  $2.2424\cdot 10^{-3}$  & $2.2205\cdot 10^{-3}$ & $ 2.2442\cdot 10^{-3}$ & $ 2.0622\cdot 10^{-3}$ & $ 2.2607\cdot 10^{-3}$  & $\mathbf{2.0336\cdot 10^{-3}}$\\
$96  $  &$2.5282\cdot 10^{-3}$ & $ 1.8273\cdot 10^{-3}$ & $ 2.3179\cdot 10^{-3}$ &  $2.0671\cdot 10^{-3}$ & $ 2.1484\cdot 10^{-3}$& $  1.8150\cdot 10^{-3}$ & $ 2.2024\cdot 10^{-3}$  & $\mathbf{1.8018\cdot 10^{-3}}$\\
$108 $& $2.5250\cdot 10^{-3} $& $ 1.6405\cdot 10^{-3} $&  $2.3580\cdot 10^{-3}$ & $ 1.9529\cdot 10^{-3}$ & $ 2.0835\cdot 10^{-3}$ &  $1.6371\cdot 10^{-3}$ & $ 2.1542\cdot 10^{-3}$ & $ \mathbf{1.6231\cdot 10^{-3}}$\\
$120 $  &$2.5211\cdot 10^{-3}$&  $ 1.4869\cdot 10^{-3}$&   $2.3664\cdot 10^{-3}$ & $ 1.8617\cdot 10^{-3} $&  $2.0206\cdot 10^{-3} $&  $1.4721\cdot 10^{-3}$&  $ 2.0960\cdot 10^{-3}$&  $\mathbf{ 1.4635\cdot 10^{-3}}$\\
$168 $ &$2.5164\cdot 10^{-3} $ & $1.1042\cdot 10^{-3}$ &  $2.1058\cdot 10^{-3}$  &$ 1.4307\cdot 10^{-3}$  & $1.8323\cdot 10^{-3} $ & $\mathbf{1.0761\cdot 10^{-3}} $ & $1.9180\cdot 10^{-3}$ & $ \mathbf{1.0761\cdot 10^{-3}}$\\
$216 $ &$ 2.5129\cdot 10^{-3}$&  $ 8.8555\cdot 10^{-4}$&  $ 1.8827\cdot 10^{-3}$  &$ 1.1820\cdot 10^{-3} $&  $1.6987\cdot 10^{-3} $&  $\mathbf{8.4862\cdot 10^{-4}}$  & $1.7668\cdot 10^{-3}$ &  $8.5065\cdot 10^{-4}$\\
$264$  &$2.5098\cdot 10^{-3}$ & $ 7.3196\cdot 10^{-4}$ &  $1.7338\cdot 10^{-3} $ & $9.7474\cdot 10^{-4}$  & $1.5784\cdot 10^{-3}$&  $ 7.0312\cdot 10^{-4}$  &$ 1.6337\cdot 10^{-3} $& $ \mathbf{6.9867\cdot 10^{-4}}$\\
$312 $&$2.5089\cdot 10^{-3}$ & $ 6.3254\cdot 10^{-4}$ & $ 1.6276\cdot 10^{-3}$ & $ 8.6721\cdot 10^{-4}$  & $1.4790\cdot 10^{-3}$ &  $6.0003\cdot 10^{-4} $ & $1.5259\cdot 10^{-3}$  &$ \mathbf{5.9896\cdot 10^{-4}}$\\
$624$  &$ 2.5054\cdot 10^{-3} $&  $3.3926\cdot 10^{-4} $& $  1.3195\cdot 10^{-3}$ & $ 5.1712\cdot 10^{-4}$ & $ 1.0687\cdot 10^{-3}$&   $\mathbf{3.1818\cdot 10^{-4}}$  & $1.0924\cdot 10^{-3}$ &  $3.1880\cdot 10^{-4}$\\
$1248$& $2.5017\cdot 10^{-3}$&   $1.8626\cdot 10^{-4}$ &  $ 1.1503\cdot 10^{-3}$ & $ 3.2887\cdot 10^{-4}$ & $ 7.1142\cdot 10^{-4}$ &  $1.7610\cdot 10^{-4}$&  $ 7.4592\cdot 10^{-4}$ &  $\mathbf{1.7548\cdot 10^{-4}}$\\
		\hline
	\end{tabular}\label{table3}
\end{table}
}

{ \begin{table}
\caption{The EMSEs of  reconstruction  from the noisy samples of $f$ and $f'$ by different methods. The standard deviation of the noise is $\sigma_\epsilon = 0.1$. The parameters of $l_1$ regularization and $l_2$ regularization are $\eta= 1.2, \alpha = 1$.} \vspace{0.2cm}	\scriptsize \setlength\tabcolsep{3pt}\centering
	\begin{tabular}{c|c c c c c c c c}
\hline
No.   &MCI  & MCI+ & Pre-filt & Pre-filt &$l_1$-Reg  &$l_1$-Reg+  &$l_2$-Reg&$l_2$-Reg+  \\
of ${S}$	&  & Post-filt &+MCI  & +MCI+ &  & Post-filt &  & Post-filt  \\
		&  &  &  &  Post-filt &  &  & & \\
\hline
$12$ &$2.0447\cdot 10^{-1}$  & $1.9696\cdot 10^{-1}$ &$1.9810\cdot 10^{-1}$ &$\mathbf{1.9187\cdot 10^{-1}}$ &$2.0228\cdot 10^{-1}$  &$2.0228\cdot 10^{-1}$   &$2.0193\cdot 10^{-1}$  &$2.0193\cdot 10^{-1}$ \\
$24$ &$1.6772\cdot 10^{-1}$ & $1.6002\cdot 10^{-1}$  &  $1.6298\cdot 10^{-1}$ &$\mathbf{1.5505\cdot 10^{-1}}$  & $1.6606\cdot 10^{-1}$ & $1.6606\cdot 10^{-1}$ & $1.6646\cdot 10^{-1}$ &$1.6646\cdot 10^{-1}$\\
$36$ & $6.3027\cdot 10^{-2}$ &  $5.8539\cdot 10^{-2}$ & $5.9992\cdot 10^{-2}$ &$\mathbf{5.7333\cdot 10^{-2}}$  & $6.1093\cdot 10^{-2}$  &  $6.1093\cdot 10^{-2}$  & $6.1223\cdot 10^{-2}$&$ 6.1223\cdot 10^{-2}$ \\
$48$ & $1.7364\cdot 10^{-2}$ & $\mathbf{1.5195\cdot 10^{-2}}$ & $1.5838\cdot 10^{-2}$ & ${1.5500\cdot 10^{-2}}$ & $1.6414\cdot 10^{-2}$ &  $1.6414\cdot 10^{-2}$ & $1.6276\cdot 10^{-2}$  &$ 1.6276\cdot 10^{-2}$\\
$60$ & $8.3509\cdot 10^{-3}$ & $\mathbf{7.0806\cdot 10^{-3}}$ & $7.4106\cdot 10^{-3}$  &$7.8499\cdot 10^{-3}$  &$7.7873\cdot 10^{-3}$  &$7.7873\cdot 10^{-3}$  & $7.8924\cdot 10^{-3}$ &$7.8924\cdot 10^{-3}$ \\
$72$ & $6.9644\cdot 10^{-3}$ &$6.0299\cdot 10^{-3}$  & $\mathbf{5.9683\cdot 10^{-3}}$ & $6.8682\cdot 10^{-3}$ & $6.5291\cdot 10^{-3}$ &$6.5291\cdot 10^{-3}$  &$6.6630\cdot 10^{-3}$ & $6.6630\cdot 10^{-3}$ \\
$84$ &$6.7706\cdot 10^{-3}$  &$6.0350\cdot 10^{-3}$  & $\mathbf{5.9369\cdot 10^{-3}}$ &$7.0210\cdot 10^{-3}$  &$6.2032\cdot 10^{-3}$   & $6.2061\cdot 10^{-3}$ &$6.3276\cdot 10^{-3}$ &$ 6.3296\cdot 10^{-3}$ \\
$96$ &$6.7236\cdot 10^{-3}$  & $5.9158\cdot 10^{-3}$ & $5.9449\cdot 10^{-3}$ & $ 6.8447\cdot 10^{-3}$ &$5.8551\cdot 10^{-3}$  &$\mathbf{5.7939\cdot 10^{-3}}$  &$6.0500\cdot 10^{-3}$ &$5.9837\cdot 10^{-3}$ \\
$108$ &$6.6940\cdot 10^{-3}$  &$5.7642\cdot 10^{-3}$  & $5.9548\cdot 10^{-3}$ & $6.6695\cdot 10^{-3}$ &$5.5846\cdot 10^{-3}$   & $\mathbf{5.4172\cdot 10^{-3}}$ &$5.8015\cdot 10^{-3}$ & $ 5.6187\cdot 10^{-3}$\\
$120$ & $6.6832\cdot 10^{-3}$ &$5.6607\cdot 10^{-3}$  & $5.8638\cdot 10^{-3}$ & $6.5043\cdot 10^{-3}$ & $5.2678\cdot 10^{-3}$ &$\mathbf{4.9812\cdot 10^{-3}}$  &$5.6014\cdot 10^{-3}$ &$5.2933\cdot 10^{-3}$ \\
$168$ &$6.6281\cdot 10^{-3}$  & $4.4436\cdot 10^{-3}$ & $5.2445\cdot 10^{-3}$ &$4.7156\cdot 10^{-3}$   & $4.5830\cdot 10^{-3}$ &$\mathbf{4.0760\cdot 10^{-3}}$  &$4.9211\cdot 10^{-3}$ &$4.2980\cdot 10^{-3}$ \\
$216$ & $6.6302\cdot 10^{-3}$ & $3.7970\cdot 10^{-3}$ &$4.9550\cdot 10^{-3}$  &$3.8718\cdot 10^{-3}$  & $4.0361\cdot 10^{-3}$ &$\mathbf{3.4458\cdot 10^{-3}}$  & $4.4431\cdot 10^{-3}$&$3.6661\cdot 10^{-3}$ \\
$264$ & $6.6175\cdot 10^{-3}$ & $3.4145\cdot 10^{-3}$ &  $4.7180\cdot 10^{-3}$& $3.5193\cdot 10^{-3}$ & $3.6399\cdot 10^{-3}$ & $\mathbf{3.0036\cdot 10^{-3}}$ &$4.0285\cdot 10^{-3}$ &$3.1937\cdot 10^{-3}$ \\
$312$ & $6.6208\cdot 10^{-3}$ &  $3.0056\cdot 10^{-3}$ & $4.5341\cdot 10^{-3}$  &$3.0455\cdot 10^{-3}$  &$3.2838\cdot 10^{-3}$  & $\mathbf{2.6363\cdot 10^{-3}}$ &$3.6855\cdot 10^{-3}$&$ 2.8286\cdot 10^{-3}$ \\
$624$ & $ 6.6386\cdot 10^{-3}$ &$1.7831\cdot 10^{-3}$  & $3.8749\cdot 10^{-3}$ & $1.8214\cdot 10^{-3}$ & $2.0782\cdot 10^{-3}$ & $\mathbf{1.5614\cdot 10^{-3}}$ &$2.4086\cdot 10^{-3}$  & $1.6609\cdot 10^{-3}$\\
$1248$ & $6.6563\cdot 10^{-3}$ & $1.0228\cdot 10^{-3}$ &$3.4142\cdot 10^{-3}$  & $1.0464\cdot 10^{-3}$ &$1.2336\cdot 10^{-3}$  &$\mathbf{8.7747\cdot 10^{-4}}$  &$1.5129\cdot 10^{-3}$ &$9.5448\cdot 10^{-4}$ \\
		\hline
	\end{tabular}\label{table4}
\end{table}
}

{ \begin{table}
\caption{The EMSEs of   reconstruction  from the noisy samples of $f$ by different methods. The standard deviation of the noise is $\sigma_\epsilon = 0.05$. The parameters of $l_1$ regularization and $l_2$ regularization are $\eta= 1.2, \alpha = 1$.} \vspace{0.2cm}	\scriptsize \setlength\tabcolsep{5pt}\centering
	\begin{tabular}{c|c c c c c c}
\hline
No.   &MCI  & MCI+   &$l_1$-Reg  &$l_1$-Reg+  &$l_2$-Reg&$l_2$-Reg+  \\
of ${S}$	&  & Post-filt  &  & Post-filt &  & Post-filt  \\
\hline
$12$   &   $1.5859\cdot 10^{-1}$&   $\mathbf{1.5599\cdot 10^{-1}}$ &  $1.5764\cdot 10^{-1}$&   $1.5764\cdot 10^{-1} $&  $1.5769\cdot 10^{-1}$  &$ 1.5769\cdot 10^{-1}$\\
$24$   & $ 7.7599\cdot 10^{-2} $&  $7.6884\cdot 10^{-2}$  & $7.7133\cdot 10^{-2} $ & $7.7133\cdot 10^{-2}$  & $\mathbf{7.6421\cdot 10^{-2}}$ & $ 7.6421\cdot 10^{-2}$\\
$36$   &  $1.9784\cdot 10^{-2}$ &  $1.8561\cdot 10^{-2}$  & $1.9135\cdot 10^{-2} $ &$ 1.9135\cdot 10^{-2}$  & $\mathbf{1.7941\cdot 10^{-2}}$  & $1.7941\cdot 10^{-2}$\\
$48 $  &  $ 5.4554\cdot 10^{-3}$ &  $4.9155\cdot 10^{-3}$ &  $5.1145\cdot 10^{-3}$ &  $5.1145\cdot 10^{-3} $&  $\mathbf{4.6144\cdot 10^{-3}}$   &$4.6144\cdot 10^{-3}$\\
$60 $  &   $2.9055\cdot 10^{-3}$  & $2.6568\cdot 10^{-3} $ & $2.7004\cdot 10^{-3}$  & $2.7004\cdot 10^{-3}$ &  $\mathbf{2.5800\cdot 10^{-3}}$  &$ 2.5800\cdot 10^{-3}$\\
$72$   &   $2.5378\cdot 10^{-3} $&  $2.3512\cdot 10^{-3}$  &$ 2.3321\cdot 10^{-3}$ &  $2.3321\cdot 10^{-3}$ &  $\mathbf{2.2939\cdot 10^{-3}}$  & $2.2939\cdot 10^{-3}$\\
$84$   &   $2.4940\cdot 10^{-3}$ &  $2.0559\cdot 10^{-3}$  & $2.2026\cdot 10^{-3}$ &  $2.0182\cdot 10^{-3}$ &  $2.2221\cdot 10^{-3}$ &  $\mathbf{2.0032\cdot 10^{-3}}$\\
$96$   &   $2.4895\cdot 10^{-3}$ & $ 1.8232\cdot 10^{-3}$ &  $2.1195\cdot 10^{-3}$ & $ 1.7891\cdot 10^{-3}$ &  $2.1736\cdot 10^{-3}$ &  $\mathbf{1.7796\cdot 10^{-3}}$\\
$108$ &  $2.4909\cdot 10^{-3} $ & $1.6392\cdot 10^{-3} $ & $2.0569\cdot 10^{-3}$  &$ 1.6031\cdot 10^{-3} $ & $2.1269\cdot 10^{-3}$ &  $\mathbf{1.5973\cdot 10^{-3}}$\\
$120$ &   $2.4902\cdot 10^{-3} $&  $1.4863\cdot 10^{-3} $ & $1.9990\cdot 10^{-3}$ & $ 1.4551\cdot 10^{-3}$  & $2.0732\cdot 10^{-3}$ &  $\mathbf{1.4431\cdot 10^{-3}}$\\
$168$ &   $2.4942\cdot 10^{-3}$ &  $1.1034\cdot 10^{-3} $  &$1.8142\cdot 10^{-3}$  &$ 1.0572\cdot 10^{-3}$ &  $1.9003\cdot 10^{-3}$&  $ \mathbf{1.0528\cdot 10^{-3}}$\\
$216$ &   $2.4956\cdot 10^{-3}$ &  $8.8030\cdot 10^{-4}$   &$1.6856\cdot 10^{-3}$  & $\mathbf{8.3641\cdot 10^{-4}} $ & $1.7557\cdot 10^{-3}$  & $8.3698\cdot 10^{-4}$\\
$264$ &  $ 2.4956\cdot 10^{-3}$  & $7.3025\cdot 10^{-4}$  & $1.5669\cdot 10^{-3}$  & $\mathbf{6.8889\cdot 10^{-4}}$  & $1.6291\cdot 10^{-3}$ &  $6.9493\cdot 10^{-4}$\\
$312$ &   $2.4968\cdot 10^{-3}$  & $6.3040\cdot 10^{-4} $  &$1.4702\cdot 10^{-3}$ & $ \mathbf{5.8854\cdot 10^{-4}}$  & $1.5213\cdot 10^{-3} $ & $5.9541\cdot 10^{-4}$\\
$624$ &   $2.4994\cdot 10^{-3}$  & $3.3780\cdot 10^{-4}$   &$1.0645\cdot 10^{-3}$  & $3.1709\cdot 10^{-4}$ &  $1.0888\cdot 10^{-3}$  & $\mathbf{3.1546\cdot 10^{-4}}$\\
$1248$ &  $ 2.4988\cdot 10^{-3}$  & $1.8447\cdot 10^{-4}$  & $7.0981\cdot 10^{-4}$ & $ 1.7487\cdot 10^{-4}$  & $7.4399\cdot 10^{-4}$  & $\mathbf{1.7397\cdot 10^{-4}}$\\
	\hline
	\end{tabular}\label{table5}
\end{table}
}

{ \begin{table}
\caption{Comparison of running time (second)   on the reconstructions from the noisy samples of $f$ and $f'$ by different methods (running $10000$ times). The standard deviation of the noise is $\sigma_\epsilon = 0.1$. The parameters of $l_1$ regularization and $l_2$ regularization are $\eta= 1.2, \alpha = 1$.} \vspace{0.2cm}	\scriptsize \scriptsize \setlength\tabcolsep{5pt}\centering
\begin{tabular}{c|c c c c c c c c}
\hline
No.  &MCI  & MCI+ & Pre-filt & Pre-filt &$l_1$-Reg  &$l_1$-Reg+  &$l_2$-Reg&$l_2$-Reg+  \\
of ${S}$	&  & Post-filt &+MCI  & +MCI+ &  & Post-filt &  & Post-filt  \\
	&  &  &  &  Post-filt &  &  & & \\
\hline
$12$ & $3.6788\cdot 10^{0}$ &  $4.3653\cdot 10^{0} $& $ 6.4205\cdot 10^{0}$ &  $8.2292\cdot 10^{0}$ &  $1.1015\cdot 10^{1}$  & $1.0952\cdot 10^{1}$  &$ 3.8853\cdot 10^{0}$ &  $4.1064\cdot 10^{0}$\\
$24$  &$ 3.8765\cdot 10^{0} $&  $5.1819\cdot 10^{0}$  & $7.1325\cdot 10^{0}$  & $8.6309\cdot 10^{0}$  & $2.1278\cdot 10^{1}$ &  $2.1175\cdot 10^{1}$ &  $4.4323\cdot 10^{0} $ &$ 4.6031\cdot 10^{0}$\\
$36$ &  $3.8860\cdot 10^{0}$ &  $6.8050\cdot 10^{0}$  & $7.8395\cdot 10^{0}$  & $1.1033\cdot 10^{1}$ &  $4.0122\cdot 10^{1}$ & $ 3.9912\cdot 10^{1}$  & $4.7568\cdot 10^{0}$  & $4.8133\cdot 10^{0}$\\
$48$ & $ 4.1999\cdot 10^{0}$  & $8.4191\cdot 10^{0}$  &$ 1.0369\cdot 10^{1}$ &  $1.5357\cdot 10^{1}$ &  $7.4549\cdot 10^{1}$ & $ 7.3722\cdot 10^{1}$ &  $5.3805\cdot 10^{0}$  & $5.4777\cdot 10^{0}$\\
$60$ &  $4.1374\cdot 10^{0}$  & $1.0196\cdot 10^{1}$  & $1.2337\cdot 10^{1}$ &  $1.8147\cdot 10^{1}$  & $1.1683\cdot 10^{2}$ &  $1.1920\cdot 10^{2} $&  $6.4216\cdot 10^{0}$  & $6.0145\cdot 10^{0}$\\
$72$ &  $3.9930\cdot 10^{0}$ &  $1.2222\cdot 10^{1}$ &  $1.4093\cdot 10^{1}$ &  $2.1907\cdot 10^{1}$  & $1.6636\cdot 10^{2}$ &  $1.6958\cdot 10^{2}$ &  $7.2450\cdot 10^{0}$  & $7.5104\cdot 10^{0}$\\
$84$ &  $4.0059\cdot 10^{0}$ &  $1.4067\cdot 10^{1}$ &  $1.5117\cdot 10^{1}$  & $2.4300\cdot 10^{1}$ &  $2.1869\cdot 10^{2}$  & $2.2381\cdot 10^{2}$  & $7.7516\cdot 10^{0}$ &  $8.2238\cdot 10^{0}$\\
$96$  & $3.9740\cdot 10^{0}$ &  $1.4904\cdot 10^{1}$  & $1.9453\cdot 10^{1}$ &  $3.0636\cdot 10^{1}$ &  $2.7540\cdot 10^{2}$ &  $2.8930\cdot 10^{2}$  & $8.8913\cdot 10^{0}$  &$ 9.4150\cdot 10^{0}$\\
$108$ & $4.0361\cdot 10^{0}$ &  $1.7351\cdot 10^{1}$ &  $2.1317\cdot 10^{1}$ &  $3.4991\cdot 10^{1}$&   $3.2556\cdot 10^{2}$  & $3.3321\cdot 10^{2}$  & $1.0162\cdot 10^{1}$&   $1.0759\cdot 10^{1}$\\
$120$  & $4.0246\cdot 10^{0} $& $ 1.8552\cdot 10^{1}$ & $ 2.3392\cdot 10^{1}$ &  $3.8498\cdot 10^{1} $&  $3.7987\cdot 10^{2}$ &  $3.9022\cdot 10^{2}$ &  $1.2139\cdot 10^{1}$ &  $1.2272\cdot 10^{1}$\\
$168$ &  $4.1341\cdot 10^{0}$ &  $1.9241\cdot 10^{1}$ &  $3.6574\cdot 10^{1}$ &  $5.3410\cdot 10^{1}$&   $6.7238\cdot 10^{2}$ &  $6.9768\cdot 10^{2} $ & $1.7152\cdot 10^{1}$ & $ 1.8839\cdot 10^{1}$\\
$216$ & $4.3097\cdot 10^{0} $&  $2.0494\cdot 10^{1}$ &  $5.9186\cdot 10^{1}$ &  $7.8245\cdot 10^{1}$ &  $1.1883\cdot 10^{3}$  & $1.2520\cdot 10^{3}$  &$ 4.1561\cdot 10^{1}$ &  $4.2246\cdot 10^{1}$\\
$264 $ & $5.3899\cdot 10^{0}$  & $2.5877\cdot 10^{1}$ & $ 9.8980\cdot 10^{1}$ &  $1.1487\cdot 10^{2} $&  $1.8915\cdot 10^{3}$ &  $1.9623\cdot 10^{3}$ &  $6.6949\cdot 10^{1}$ &  $6.9506\cdot 10^{1}$\\
$312$  & $5.4321\cdot 10^{0}$ &  $2.7976\cdot 10^{1}$ &  $1.2620\cdot 10^{2}$  &$ 1.4583\cdot 10^{2} $ & $2.8349\cdot 10^{3}$ &  $2.9861\cdot 10^{3} $&  $1.1200\cdot 10^{2}$ &  $1.0305\cdot 10^{2}$\\
$624$  & $1.0396\cdot 10^{1}$ &  $4.5142\cdot 10^{1}$ &  $7.1446\cdot 10^{2}$  & $7.2646\cdot 10^{2}$ &  $1.2915\cdot 10^{4}$ &  $1.2899\cdot 10^{4}$  &$ 5.1578\cdot 10^{2}$  & $5.0876\cdot 10^{2}$\\
$1248$ &  $2.3852\cdot 10^{1}$ &  $8.2320\cdot 10^{1}$ &  $3.4522\cdot 10^{3}$ &  $3.5232\cdot 10^{3}$ &  $3.4581\cdot 10^{4}$ &  $3.4218\cdot 10^{4}$ &  $2.9432\cdot 10^{3}$  & $3.1514\cdot 10^{3}$\\
	\hline
	\end{tabular}\label{table6}
\end{table}
}

\begin{figure}
  \centering
  \includegraphics[width=15.9cm]{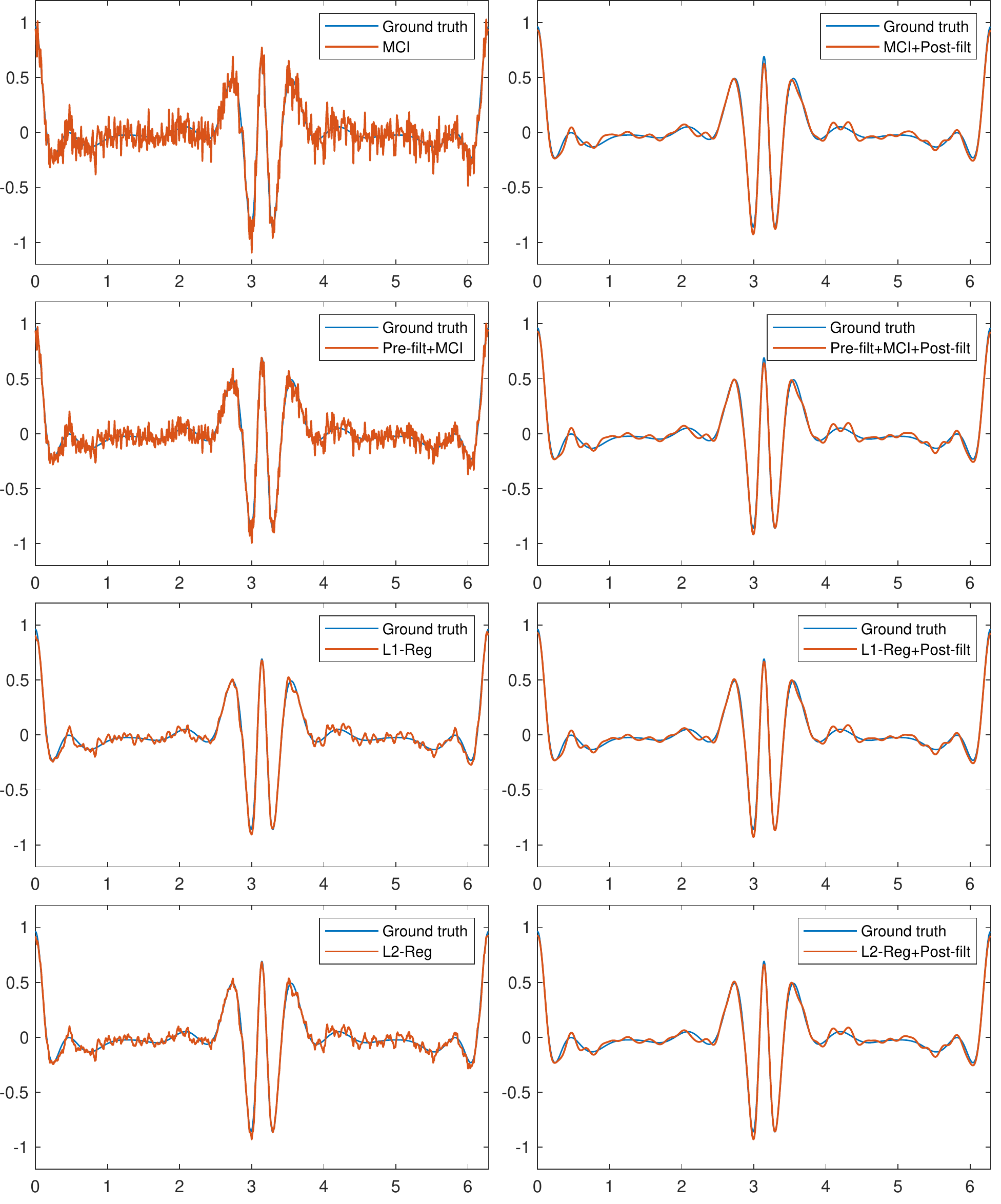}
  \caption{The reconstructed results (red lines) by different methods under the sampling  scheme of FD2 in the noisy environment with $N_s=1248, \sigma_{\epsilon}=0.1$; the blue line is the  original test function $f(t)$. The parameters of $l_1$ and $l_2$ regularization are $\eta= 1.2, \alpha = 1$.}\label{Compare_results}
\end{figure}

\begin{figure}
  \centering
  \includegraphics[width=15cm]{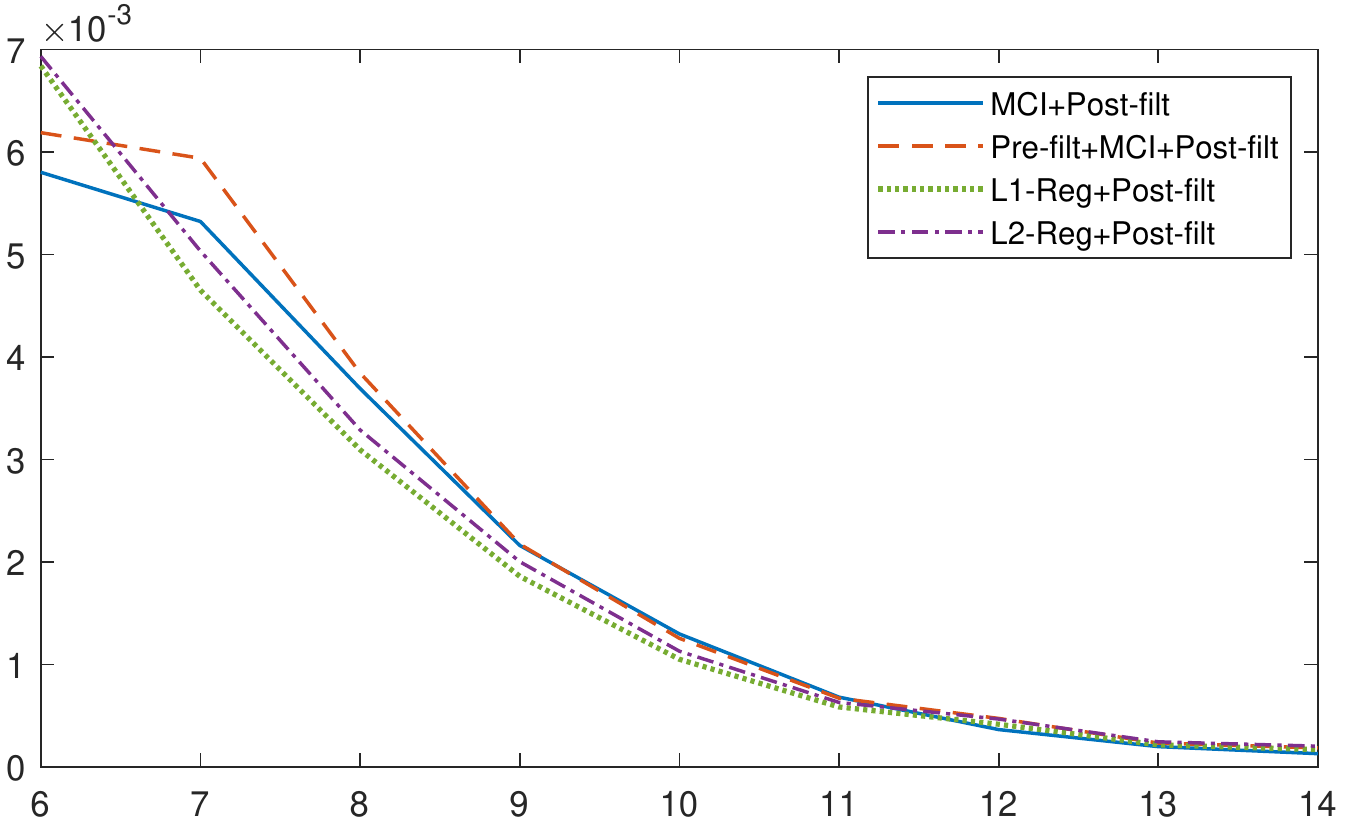}
  \caption{Experimental convergence analysis. The horizontal axis  represents   $\log_2 N_s$, where $N_s$ is the total number of samples. The vertical axis represents the error of reconstruction.}\label{Convergent_of_Post_filt}
\end{figure}

~
\newpage
~


\begin{thebibliography}{10}
\providecommand{\url}[1]{#1}
\csname url@samestyle\endcsname
\providecommand{\newblock}{\relax}
\providecommand{\bibinfo}[2]{#2}
\providecommand{\BIBentrySTDinterwordspacing}{\spaceskip=0pt\relax}
\providecommand{\BIBentryALTinterwordstretchfactor}{4}
\providecommand{\BIBentryALTinterwordspacing}{\spaceskip=\fontdimen2\font plus
\BIBentryALTinterwordstretchfactor\fontdimen3\font minus
  \fontdimen4\font\relax}
\providecommand{\BIBforeignlanguage}[2]{{%
\expandafter\ifx\csname l@#1\endcsname\relax
\typeout{** WARNING: IEEEtran.bst: No hyphenation pattern has been}%
\typeout{** loaded for the language `#1'. Using the pattern for}%
\typeout{** the default language instead.}%
\else
\language=\csname l@#1\endcsname
\fi
#2}}
\providecommand{\BIBdecl}{\relax}
\BIBdecl

\bibitem{papoulis1977generalized}
A.~Papoulis, ``Generalized sampling expansion,'' \emph{{IEEE} Trans. Circuits
  Syst.}, vol.~24, no.~11, pp. 652--654, 1977.

\bibitem{Grochenig2020sharp}
K.~Gröchenig, J.~Romero, and J.~Stöckler, ``Sharp results on sampling with
  derivatives in shift-invariant spaces and multi-window {G}abor frames,''
  \emph{Constructive Approximation}, vol.~51, no.~1, 2020.

\bibitem{monich2017two}
U.~J. M{\"o}nich and H.~Boche, ``A two channel system approximation for
  bandlimited functions,'' \emph{IEEE Transactions on Information Theory},
  vol.~63, no.~9, pp. 5496--5505, 2017.

\bibitem{liu2017signal}
N.~Liu, R.~Tao, R.~Wang, Y.~Deng, N.~Li, and S.~Zhao, ``Signal reconstruction
  from recurrent samples in fractional {Fourier} domain and its application in
  multichannel {SAR},'' \emph{Signal Processing}, vol. 131, pp. 288--299, 2017.

\bibitem{xu2017multichannel}
L.~Xu, R.~Tao, and F.~Zhang, ``Multichannel consistent sampling and
  reconstruction associated with linear canonical transform,'' \emph{IEEE
  Signal Processing Letters}, vol.~24, no.~5, pp. 658--662, 2017.

\bibitem{Shah2021Lattice}
F.~A. Shah and A.~Y. Tantary, ``Lattice-based multi-channel sampling theorem
  for linear canonical transform,'' \emph{Digital Signal Processing}, vol. 117,
  p. 103168, 2021.

\bibitem{wei2019convolution}
D.~Wei and Y.~M. Li, ``Convolution and multichannel sampling for the offset
  linear canonical transform and their applications,'' \emph{IEEE Trans. Signal
  Process.}, vol.~67, no.~23, pp. 6009--6024, 2019.

\bibitem{cheng2020Multicnon}
D.~Cheng and K.~I. Kou, ``Multichannel interpolation of nonuniform samples with
  application to image recovery,'' \emph{J. Comput. Appl. Math.}, vol. 367, p.
  112502, 2020.

\bibitem{wigderson2021uncertainty}
A.~Wigderson and Y.~Wigderson, ``The uncertainty principle: variations on a
  theme,'' \emph{Bulletin of the American Mathematical Society}, vol.~58,
  no.~2, pp. 225--261, 2021.

\bibitem{xiao2013sampling}
L.~Xiao and W.~Sun, ``Sampling theorems for signals periodic in the linear
  canonical transform domain,'' \emph{Opt. Commun.}, vol. 290, pp. 14--18,
  2013.

\bibitem{Mohammadi2018sampling}
E.~Mohammadi and F.~Marvasti, ``Sampling and distortion tradeoffs for
  bandlimited periodic signals,'' \emph{IEEE Transactions on Information
  Theory}, vol.~64, no.~3, pp. 1706--1724, 2018.

\bibitem{cheng2019FFTMCI}
D.~Cheng and K.~I. Kou, ``{FFT} multichannel interpolation and application to
  image super-resolution,'' \emph{Signal Process.}, vol. 162, pp. 21 -- 34,
  2019.

\bibitem{fraser1989interpolation}
D.~Fraser, ``Interpolation by the {FFT} revisited-an experimental
  investigation,'' \emph{IEEE Transactions on Acoustics, Speech, and Signal
  Processing}, vol.~37, no.~5, pp. 665--675, 1989.

\bibitem{Pawlak2003postfilter}
M.~{Pawlak}, E.~{Rafajlowicz}, and A.~{Krzyzak}, ``Postfiltering versus
  prefiltering for signal recovery from noisy samples,'' \emph{IEEE Trans. Inf.
  Theory}, vol.~49, no.~12, pp. 3195--3212, 2003.

\bibitem{cvx}
M.~Grant and S.~Boyd, ``{CVX}: Matlab software for disciplined convex
  programming, version 2.1,'' \url{http://cvxr.com/cvx}, Mar. 2014.

\bibitem{chen2006new}
J.~Chen, J.~Benesty, Y.~Huang, and S.~Doclo, ``New insights into the noise
  reduction {W}iener filter,'' \emph{IEEE Trans. Audio Speech Lang. Process.},
  vol.~14, no.~4, pp. 1218--1234, 2006.

\bibitem{Ramani2008nonideal}
S.~{Ramani}, D.~{Van De Ville}, T.~{Blu}, and M.~{Unser}, ``Nonideal sampling
  and regularization theory,'' \emph{IEEE Trans. Signal Process.}, vol.~56,
  no.~3, pp. 1055--1070, March 2008.

\bibitem{donoho2006stable}
D.~Donoho, M.~Elad, and V.~Temlyakov, ``Stable recovery of sparse overcomplete
  representations in the presence of noise,'' \emph{IEEE Transactions on
  Information Theory}, vol.~52, no.~1, pp. 6--18, 2006.

\bibitem{Ohlsson2010}
H.~Ohlsson, L.~Ljung, and S.~Boyd, ``Segmentation of {ARX}-models using
  sum-of-norms regularization,'' \emph{Automatica}, vol.~46, no.~6, pp.
  1107--1111, 2010.

\bibitem{Boyd2011}
S.~Boyd, N.~Parikh, E.~Chu, B.~Peleato, and J.~Eckstein, ``Distributed
  optimization and statistical learning via the alternating direction method of
  multipliers,'' \emph{Foundations and Trends® in Machine Learning}, vol.~3,
  no.~1, pp. 1--122, 2011.

\bibitem{zou2018robust}
C.~Zou and K.~I. Kou, ``Robust signal recovery using the prolate spherical wave
  functions and maximum correntropy criterion,'' \emph{Mech. Syst. Signal
  Proc.}, vol. 104, pp. 279--289, 2018.

\end{thebibliography}
\end{document}